\documentclass[journal,onecolumn]{IEEEtran}

\usepackage{mathrsfs}

\usepackage{calligra}
\usepackage{amsfonts,amscd,mathrsfs,amsmath,amsthm,amssymb}
\usepackage{mathtools}
\usepackage{xparse}
\usepackage{graphicx}
\usepackage{float}
\usepackage{pifont}
\usepackage[dvipsnames]{xcolor}
\usepackage{colortbl}
\usepackage{multirow}
\usepackage{enumerate}   
\usepackage{comment}
\usepackage{booktabs}
\usepackage{cancel}
\usepackage{lieart}
\usepackage{mhchem}
\usepackage{listings}
\usepackage{xcolor}
\usepackage{physics}
\usepackage{bm}
\usepackage{bbm}
\usepackage{caption}
\usepackage{lipsum}

\PassOptionsToPackage{hyphens}{url}\usepackage{hyperref}
\hypersetup{
    colorlinks=true,
    linkcolor=magenta!80!black,  
    citecolor=magenta!80!black,
    breaklinks=true,
}
\usepackage[capitalise]{cleveref}

\theoremstyle{theorem}

\theoremstyle{definition}
\newtheorem{theorem}{Theorem}
\newtheorem*{theorem*}{Theorem}

\newtheorem*{conjecture*}{Conjecture}
\newtheorem{corollary}{Corollary}
\newtheorem*{corollary*}{Corollary}
\newtheorem{lemma}{Lemma}
\newtheorem*{lemma*}{Lemma}
\newtheorem{proposition}{Proposition}
\newtheorem{remark}{Remark}

\AddToHook{cmd/appendix/before}{%
  \setcounter{equation}{0}%
  \renewcommand{\theequation}{\thesection\arabic{equation}}%
  \setcounter{theorem}{0}%
  \renewcommand{\thetheorem}{\thesection\arabic{theorem}}%
  \setcounter{lemma}{0}%
  \renewcommand{\thelemma}{\thesection\arabic{lemma}}%
}

\newcommand\numberthis{\addtocounter{equation}{1}\tag{\theequation}}

\DeclareMathAlphabet{\mathsfit}{T1}{\sfdefault}{\mddefault}{\sldefault}
\SetMathAlphabet{\mathsfit}{bold}{T1}{\sfdefault}{\bfdefault}{\sldefault}

\newcommand{\CC}{\mathbb{C}}

\renewcommand{\C}{\mathcal{C}} 
\newcommand{\Cint}{\mathcal{I}}

\newcommand{\Ad}{\textsf{Ad}}
\newcommand{\Sym}{\mathrm{Sym}}

\newcommand{\Twirl}{\text{Twirl}}
\newcommand{\Hom}{\text{Hom}}

\newcommand{\lsuper}{\langle\!\langle}
\newcommand{\rsuper}{\rangle\!\rangle}

\newcommand{\lcode}{\{\mkern-4mu\{}
\newcommand{\rcode}{\}\mkern-4mu\}}

\renewcommand{\L}{\mathscr{L}} 
\renewcommand{\E}{W} 
\renewcommand{\B}{\mathcal{B}} 

\newcommand{\proj}{\mathscr{P}} 
\newcommand{\twirl}{\mathscr{T}} 

\renewcommand{\ip}[1]{\expval{#1}} 
\newcommand{\ipp}[1]{\expval{\! \expval{#1} \!}} 

\renewcommand{\op}{\mathscr{O}} 

\newcommand{\Avec}{\bm{A}}
\newcommand{\Bvec}{\bm{B}}

\providecommand{\SU}[1]{\mathrm{SU}(#1)}
\begin{document}

\title{MacWilliams Identities for Intrinsic Quantum Codes}

\author{Eric~Kubischta, Ian~Teixeira

\thanks{Both authors contributed equally to this work.}
\thanks{Eric Kubischta is affiliated with the Department of Mathematics, Florida State University, Tallahassee, FL 32306.}
\thanks{Ian Teixeira is affiliated with the Department of Mathematics, University of California, San Diego, CA 92093}
}

\maketitle
\begin{abstract}
This paper develops a MacWilliams-type enumerator theory and semidefinite programming bounds for intrinsic quantum codes. An intrinsic code is a subspace $\C\subseteq V$ of a finite-dimensional unitary representation $V$ of a group $G$, with errors organized into the isotypic components of $\L(V)$ under conjugation. The Knill-Laflamme conditions persist under every \(G\)-equivariant realization of \(V\), allowing a single intrinsic bound to apply across a broad class of physical encodings. Without assuming a group action, we first construct projector and twirl quadratic enumerators for an arbitrary orthogonal decomposition of $\L(V)$, analogous respectively to the $A$- and $B$-enumerators of Shor--Laflamme theory. These enumerators satisfy positivity, normalization, and Knill--Laflamme inequalities, but the two superoperator families need not span the same space and therefore admit no canonical duality. A group action places their generalized versions in the same equivariant intertwiner algebra, and the change of basis between them defines the intrinsic MacWilliams transform. When the conjugation representation has multiplicities, this algebra is noncommutative and the enumerators are matrix-valued; positivity and detection are encoded by Loewner-order constraints, making the natural programming bound a semidefinite program. The multiplicity-free specialization recovers the known scalar MacWilliams transform and linear-programming framework. For symmetric powers, intrinsic depth equals ordinary qudit distance, yielding extremality certificates for permutation-invariant qudit codes. Finally, for the $27$-dimensional $\SU3$ representation $V=(2,2)$, an exact dual certificate proves the sharp bound $K\leq5$ for codes detecting the full adjoint isotypic component. An explicit $5$-dimensional code attains the bound and, under the canonical embedding, gives a four-qutrit distance-$2$ code. 
\end{abstract}


\tableofcontents

\section{Introduction}
\label{sec:introduction}

\subsection{Overview}

A central problem in coding theory is to bound the parameters of
error-correcting codes.
In both the classical and quantum settings, the sharpest such bounds come from
enumerator formalisms together with duality relations.
In the classical theory, the MacWilliams identities \cite{MacWilliams1963}
relate the weight distribution of a code to that of its dual and, with
positivity, yield Delsarte's linear programming bounds \cite{Delsarte1973} and
their association-scheme refinements
\cite{MacWilliamsSloane1977,BannaiIto1984}.
In quantum coding theory, the Shor--Laflamme weight enumerators play the same
role for stabilizer and nonstabilizer codes
\cite{GF4codes,rainsWE,ShorWE,Shadows}.
The pattern is always the same.
One attaches enumerators to a code, proves linear relations among them, and
combines those relations with positivity to obtain an optimization bound.

Recent work introduced a symmetry-based viewpoint through the notion of an
\emph{intrinsic quantum code} \cite{usIntrinsicCodes}.
Here a code is not a subspace of a fixed tensor-product Hilbert space, but a
subspace
\[
    \C \subseteq V
\]
of a finite-dimensional unitary representation \(V\) of a group \(G\).
The operator space \(\L(V):=\operatorname{End}(V)\) is a \(G\)-representation
under conjugation, and errors are organized by its decomposition into intrinsic
error sectors.
The Knill--Laflamme conditions become conditions on these sectors.
A single intrinsic code then governs many physical realizations at once, so a
bound proved intrinsically applies to a whole family of ordinary codes sharing
the same symmetry data.
When an intrinsic depth has been fixed, we write
\[
    \lcode N,K,d\rcode_G
\]
for a \(K\)-dimensional intrinsic code, detecting all sectors of depth less
than \(d\), inside a representation \(V\) of dimension \(N:=\dim V\); the first
slot records the dimension of the ambient representation. Intrinsic codes and
error sectors are introduced in
Section~\ref{sec:intrinsic-forms} following \cite{usIntrinsicCodes}; there we also define an abstract
sector-depth assignment as part of the chosen error model.
The physically meaningful depth assignments used for symmetric powers, and
their identification with ordinary qudit distance, are developed in
Appendix~\ref{sec:mf-examples}.

This paper develops enumerator theory for intrinsic codes and uses it to bound
their parameters.
The organizing principle is a single dichotomy.
The representation theory of the conjugation action on \(\L(V)\) fixes not
only the error sectors, but also the form of the duality transform and the
class of the resulting optimization problem.
When the conjugation representation is multiplicity-free, the intertwiner
algebra is commutative, the enumerators are scalars, the MacWilliams transform
is a unitary matrix, and the bound is a linear program.
When multiplicities are present, the intertwiner algebra is noncommutative,
the enumerators are matrices, the MacWilliams transform is a unitary change of
basis on the full intertwiner space, and the bound is a semidefinite program.
The linear-programming case is the commutative shadow of the general
matrix-valued theory.

\subsection{Main contribution: the semidefinite case}

The multiplicity-free case has already been substantially developed.
For quantum metric spaces with sufficient symmetry it is the content of
Okada's linear programming theory
\cite{okada2025quantumanalogdelsarteslinear}, and for the spin case it goes
back to Bumgardner \cite{bumgardner2012codeswastmetricspacestheory}.
In that setting the error sectors behave like the classes of an association
scheme, the projector and twirl maps furnish dual bases of a commutative
intertwiner algebra, and the code constraints are linear.
The multiplicity-free formalism in
Appendix~\ref{sec:mf-macwilliams} and Appendix~\ref{sec:mf-examples} is therefore a
reformulation of this commutative theory in the intrinsic projector/twirl
language.
We include it because it fixes notation, makes the connection with ordinary
quantum weight enumerators explicit, and supplies the commutative model that
the matrix-valued theory extends. 
We discuss the relationship in detail in
Section~\ref{subsec:relation-to-prior}.

The main contribution of this paper is the theory beyond the multiplicity-free
case.
Multiplicity-freeness is a strong restriction.
It holds for symmetric powers and several other important families, but many
irreducible representations \(V\) have nontrivial multiplicities in the
conjugation representation
\[
    \L(V)\cong V\otimes V^* ,
\]
and so fall outside the scalar linear-programming framework.
The new issue is not merely that an irreducible label occurs more than once:
coherent combinations of equivalent copies are themselves error directions,
and controlling them requires the full multiplicity-space matrix rather than
only scalar data attached to individual copies.
When this happens the intertwiner algebra
\[
    \Hom_G(\L(V),\L(V))
    \cong
    \bigoplus_\xi \CC^{m_\xi\times m_\xi}
\]
is noncommutative.
Scalar enumerators can only record the diagonal part of each block and miss
the off-diagonal intertwiners between distinct copies of an irreducible sector.
The correct enumerators are therefore matrix-valued, indexed by the
multiplicity spaces, and the projector and twirl constructions give two
orthonormal bases of the intertwiner algebra: a projector matrix-unit basis
and a twirl-derived orthogonal basis.
The change of basis between them is a unitary matrix-valued MacWilliams
transform on the full intertwiner space.
Positivity and detection become Loewner-order inequalities, and the feasibility
problem is a semidefinite program (SDP).

This SDP is not introduced as a strengthening of the usual scalar LP as in
\cite{LaiTsengYu,MunneNemecHuber}, where semidefinite methods sharpen bounds
for ordinary quantum codes.
Here the semidefinite structure is forced by the noncommutativity of the
intertwiner algebra, and it is the natural home of the enumerator constraints
whenever \(\L(V)\) has multiplicities.
We develop this matrix-valued theory in Section~\ref{sec:mult} and illustrate
it in Section~\ref{sec:mult-su3-example} with an explicit non-multiplicity-free
example, the \(\SU 3\) irrep of highest weight \((2,2)\), where an exact positive
dual certificate for the matrix-valued MacWilliams constraints proves the sharp
depth-\(2\) bound \(K\le5\), attained by a \(5\)-dimensional \(3.A_6\)
constituent.

\subsection{Relation to prior work}
\label{subsec:relation-to-prior}

For the scalar enumerator formalism in the multiplicity-free case, the linear programming bounds, and the
quantum Hamming-space recovery of the Shor--Laflamme and Rains enumerators are
already present in Okada's quantum-metric-space theory
\cite{okada2025quantumanalogdelsarteslinear}.
The \(\SU 2\) transform in terms of Wigner \(6j\)-symbols
(Section~\ref{subsec:su2-macwilliams}) goes back to Bumgardner
\cite{bumgardner2012codeswastmetricspacestheory} and also appears in
\cite{okada2025quantumanalogdelsarteslinear}.
The symmetric-power family \(V=\Sym^n(\CC^q)\) appears there as well, and the
distance-\(2\) enumerator characterizations we obtain in
Sections~\ref{sec:su2-example}--\ref{sec:su3-example} recover the distance-\(2\)
bounds of \cite{okada2025quantumanalogdelsarteslinear}.

We reprove these results in the intrinsic language for self-containedness and
because the intrinsic derivation, through the projector/twirl MacWilliams
transform, is the commutative specialization of the matrix-valued theory that
follows.
On the multiplicity-free side, we also add one coding-theoretic ingredient not
present in \cite{okada2025quantumanalogdelsarteslinear}.
That framework works entirely within the abstract quantum metric, where
distance is measured by the adjoint-order filtration and is not identified
with physical qudit weight.
For symmetric powers we make this identification explicit: there is a
distance-preserving correspondence between intrinsic codes in
\(\Sym^n(\CC^q)\) and permutation-invariant codes on \(n\) qudits, under which
intrinsic depth equals ordinary qudit distance
(Proposition~\ref{prop:PI-qudit-intrinsic-strong} in
Appendix~\ref{sec:mf-examples}, building on \cite{us2,usIntrinsicCodes}).
Through it the intrinsic linear program yields extremality certificates for
concrete permutation-invariant codes: the four-qubit \(((4,2,2))\), the
seven-qubit \(((7,2,3))\), and the three-qutrit \(((3,2,2))_3\) codes are each
shown to be optimal within the permutation-invariant family.

Bradshaw, LaBorde, and Montero~\cite{Bradshaw_2026} develop a complementary
representation-theoretic framework for quantum codes beyond the Pauli
stabilizer setting. Starting from a finite-group representation, their
construction produces group-invariant code spaces, interprets part of the
group action as passive error mitigation against errors in the image of the
representation, and uses projective measurements onto isotypic components as a
symmetry-resolved analogue of syndrome diagnosis; ordinary stabilizer codes
arise as a special case. The present work asks a different question. Given a
symmetry-decomposed operator space, we construct enumerators and MacWilliams
duality relations and use their positivity to bound code dimension. The two
approaches thus share the use of representation theory to organize errors, but
their principal outputs are complementary: symmetry-based code construction and
diagnosis on one side, and enumerator duality and optimization bounds on the
other. The distinction is sharpest in the presence of multiplicity, where our
variables live on the full multiplicity spaces and the resulting constraints
are semidefinite.

\subsection{Shor--Laflamme enumerators as a special case}

The intrinsic formalism contains the ordinary quantum weight enumerators as its
simplest instance, which orients the theory for the reader familiar with the
Shor--Laflamme picture \cite{ShorWE}.
Take
\[
    V=(\CC^q)^{\otimes n},
    \qquad
    G=U(q)\wr S_n,
\]
with \(G\) acting by local unitaries and permutations.
On one site, conjugation splits
\[
    \L(\CC^q)\cong \mathbf 1\oplus \Ad
\]
into the scalars and the traceless adjoint sector.
Passing to the wreath product merges supports of equal size, so the intrinsic
error sectors are the Hamming-weight spaces \(\E_w\), and the normalized
intrinsic enumerators are exactly the Shor--Laflamme enumerators:
\[
    \widetilde A_w=A_w^{\mathrm{SL}},
    \qquad
    \widetilde B_w=B_w^{\mathrm{SL}}.
\]
The one-site transform
\[
    M_{\mathrm{site}}
    =
    \frac{1}{q}
    \begin{pmatrix}
        1 & 1\\
        q^2-1 & -1
    \end{pmatrix}
\]
tensors and passes to the weight basis to give the quantum Krawtchouk
transform, and the intrinsic MacWilliams identity becomes the ordinary
Shor--Laflamme MacWilliams identity.
The details are in Sections~\ref{sec:sl-recovery}
and~\ref{subsec:wreath-krawtchouk}.
Thus the Shor--Laflamme theory is the wreath-product instance of the intrinsic
theory.

\section{General Quadratic Enumerators}
\label{sec:intrinsic-forms}

Let \(V\) be a finite-dimensional Hilbert space, and let \(\L(V)\) denote the
space of linear operators on \(V\).
We write
\[
    N:=\dim V.
\]
A code in \(V\) is a subspace \(\C\subseteq V\), and we write \(P\) for its
orthogonal projector.
The dimension of the code is denoted
\[
    K:=\dim \C.
\]

We equip \(\L(V)\) with the Hilbert--Schmidt inner product
\[
    \ip{X_1,X_2} := \Tr(X_1^\dagger X_2),
    \numberthis
\]
which is conjugate-linear in the first argument and linear in the second.
Throughout this section we fix an orthogonal decomposition
\begin{equation}
    \L(V)=\bigoplus_{\xi\in\mathcal E}\E_\xi
    \label{eq:operator-decomposition}
\end{equation}
with respect to the Hilbert--Schmidt inner product.
The finite set \(\mathcal E\) is the set of error-sector labels, and
\(\E_\xi\) is the corresponding error sector.

The word ``sector'' should be read broadly in this section.
We do not yet assume that the decomposition
\eqref{eq:operator-decomposition} comes from a group action.
It is simply a chosen orthogonal decomposition of the error space \(\L(V)\).
In the equivariant sections below, the decomposition will be the isotypic
decomposition for the conjugation action of \(G\) on \(\L(V)\); in the ordinary
\(n\)-qudit case, it will become the usual Hamming-weight decomposition of
Pauli/Weyl errors.
Thus the reader may keep the ordinary Shor--Laflamme decomposition by error
weight in mind as the guiding example.

We say that an intrinsic code \(\C\) detects the sector \(\E_\xi\) if
\[
    PEP=c_E P
    \qquad
    \text{for every }E\in\E_\xi
    \numberthis
\]
for suitable scalars \(c_E\in\CC\). This is an intrinsic Knill-Laflamme condition. Equivalently, it suffices to check this condition on any basis of \(\E_\xi\).
A depth function on the sector decomposition is a function
\[
    \text{depth}:\mathcal E\to \mathbb Z_{\geq 0}.
\]
Given such a function, a code has sector distance at least \(d\) if it detects
every sector \(\E_\xi\) with \( \text{depth}(\xi)<d\).
In this paper the depth function is part of the chosen error model.
In the ordinary tensor-product case it is Hamming weight, and for symmetric
powers it agrees with ordinary qudit distance under the correspondence with
permutation-invariant quantum codes discussed later.

This decomposition provides the basic setting for the intrinsic enumerator
formalism developed in the paper.
From it we define quadratic forms associated with the orthogonal projections
onto the sectors \(\E_\xi\) and with the corresponding twirling operations, and
we establish the positivity, normalization, and Knill--Laflamme type
inequalities satisfied by these forms.
At this level of generality, however, there is no canonical relation between
the resulting enumerators.
The MacWilliams-type identities needed for programming bounds arise only after
additional symmetry is imposed, which will be done in the later sections.

\subsection{Projectors and Twirling Maps}

We denote superoperators, i.e.\ linear maps on \(\L(V)\), by script letters.
For each \(\xi\in\mathcal E\), let
\[
    \proj_\xi:\L(V)\to\L(V)
\]
denote the orthogonal projector onto the subspace \(\E_\xi\).
These superoperators satisfy
\[
    \proj_\xi\proj_\eta=\delta_{\xi\eta}\proj_\xi,
    \numberthis
\]
so they form a family of orthogonal idempotents in \(\L(\L(V))\), the space of
superoperators acting on \(\L(V)\).

We use the Hilbert--Schmidt inner product on \(\L(\L(V))\):
\[
    \ipp{\Phi,\Psi}
    :=
    \Tr_{\L(V)}(\Phi^\dagger\Psi).
    \numberthis
\]
Equivalently, if \(\mathcal M\) is any Hilbert--Schmidt orthonormal basis of
\(\L(V)\), then
\[
    \ipp{\Phi,\Psi}
    =
    \sum_{M\in\mathcal M}\ip{\Phi(M),\Psi(M)}.
    \numberthis
\]
Thus the trace is the ordinary matrix trace of the induced linear maps on the
Hilbert space \(\L(V)\).

With respect to this inner product we have
\[
    \ipp{\proj_\xi,\proj_\eta}
    =
    \Tr(\proj_\xi^\dagger\proj_\eta)
    =
    d_\xi \,\delta_{\xi\eta},
    \numberthis
\]
where \(d_\xi=\dim\E_\xi\).
Consequently,
\[
    \left\{ \frac{1}{\sqrt{d_\xi}}\proj_\xi\right\}_{\xi\in\mathcal E}
    \numberthis
\]
is an orthonormal set in \(\L(\L(V))\).

Next fix an orthonormal basis \(\B_\xi\) of \(\E_\xi\).
We define the associated twirling superoperator
\[
    \twirl_\xi:\L(V)\to\L(V)
\]
by \cite{TwirlOrigins,TwirlOriginsFormal}
\begin{equation}
    \twirl_\xi(X):=\sum_{E\in\B_\xi}E^\dagger X E.
    \label{eq:twirl-def}
\end{equation}

\begin{lemma}
The definition of \(\twirl_\xi\) in \eqref{eq:twirl-def} is independent of the
choice of orthonormal basis of \(\E_\xi\).
\end{lemma}

\begin{proof}
Let \(\{E_j\}\) and \(\{E_j'\}\) be two orthonormal bases of \(\E_\xi\).
Then there exists a unitary matrix \(Q=(q_{jk})\) such that
\[
    E_j'=\sum_k q_{jk}E_k.
    \numberthis
\]
Substituting this expansion into \eqref{eq:twirl-def} gives
\begin{align*}
\sum_j (E_j')^\dagger X E_j'
&=
\sum_j
\left(\sum_k \overline{q_{jk}}\,E_k^\dagger\right)
X
\left(\sum_\ell q_{j\ell}E_\ell\right) \\
&=
\sum_{k,\ell}
\left(\sum_j \overline{q_{jk}}\,q_{j\ell}\right)
E_k^\dagger X E_\ell.
\end{align*}
Since \(Q\) is unitary, we have
\[
    \sum_j \overline{q_{jk}}q_{j\ell}=\delta_{k\ell},
    \numberthis
\]
and therefore
\[
    \sum_j (E_j')^\dagger X E_j'
    =
    \sum_k E_k^\dagger X E_k.
    \numberthis
\]
Hence \(\twirl_\xi\) does not depend on the choice of orthonormal basis.
\end{proof}

Let \(\B\) be the union of the mutually orthogonal bases \(\B_\xi\), relabeled
as a single Hilbert--Schmidt orthonormal basis of \(\L(V)\).

\begin{lemma}[Completeness identity]
For any Hilbert--Schmidt orthonormal basis \(\B\) of \(\L(V)\),
\[
    \sum_{E\in\B}E^\dagger X E=\Tr(X)\,I
    \qquad
    \text{for all }X\in\L(V).
\]
\end{lemma}

\begin{proof}
The map
\[
    \Phi(X):=\sum_{E\in\B}E^\dagger X E
    \numberthis
\]
is independent of the choice of Hilbert--Schmidt orthonormal basis.
It therefore suffices to compute it using the matrix-unit basis
\(\{|a\rangle\langle b|\}_{a,b}\) of \(\L(V)\).
Then
\begin{align}
\sum_{a,b} (|a\rangle\langle b|)^\dagger X (|a\rangle\langle b|)
&=
\sum_{a,b} |b\rangle\langle a| X |a\rangle\langle b| \\
&=
\sum_{a,b} \langle a|X|a\rangle\,|b\rangle\langle b| \\
&=
\Tr(X)\,I .
\end{align}
\end{proof}

Using this identity we can compute the inner products of the twirling maps.

\begin{lemma}
For all \(\xi,\eta\in\mathcal E\), we have
\[
    \ipp{\twirl_\xi,\twirl_\eta}=d_\xi\,\delta_{\xi\eta}.
\]
Consequently
\[
    \left\{\frac{1}{\sqrt{d_\xi}}\twirl_\xi\right\}_{\xi\in\mathcal E}
    \numberthis
\]
is an orthonormal set in \(\L(\L(V))\).
\end{lemma}

\begin{proof}
Using the Hilbert--Schmidt inner product on \(\L(\L(V))\) and the basis
\(\B\) of \(\L(V)\), we obtain
\begin{align}
\ipp{\twirl_\xi,\twirl_\eta}
&=
\sum_{M\in\B}\ip{\twirl_\xi(M),\twirl_\eta(M)} \\
&=
\sum_{M\in\B}\sum_{E\in\B_\xi}\sum_{F\in\B_\eta}
\Tr\!\left((E^\dagger M E)^\dagger F^\dagger M F\right).
\end{align}
Rearranging the sums gives
\[
\ipp{\twirl_\xi,\twirl_\eta}
=
\sum_{E\in\B_\xi}\sum_{F\in\B_\eta}
\sum_{M\in\B}
\Tr\!\left(E^\dagger M^\dagger E F^\dagger M F\right).
\numberthis
\]
Applying the completeness identity to the orthonormal basis
\(\{M^\dagger:M\in\B\}\) with \(X=EF^\dagger\) yields
\[
    \sum_{M\in\B} M(EF^\dagger)M^\dagger
    =
    \Tr(EF^\dagger)\,I.
    \numberthis
\]
Substituting this into the previous expression gives
\begin{align}
\ipp{\twirl_\xi,\twirl_\eta}
&=
\sum_{E\in\B_\xi}\sum_{F\in\B_\eta}
\Tr(EF^\dagger)\Tr(E^\dagger F) \\
&=
\sum_{E\in\B_\xi}\sum_{F\in\B_\eta}
|\ip{E,F}|^2 .
\end{align}
Since the sets \(\B_\xi\) are orthonormal and mutually orthogonal, the sum
vanishes unless \(\xi=\eta\), in which case it equals \(d_\xi\).
\end{proof}

At this level of generality there is no reason for the two families
\(\{\proj_\xi\}_{\xi\in\mathcal E}\) and
\(\{\twirl_\xi\}_{\xi\in\mathcal E}\) to span the same subspace of
\(\L(\L(V))\).
In particular, neither family need form a basis of \(\L(\L(V))\), and there is
therefore no canonical change-of-basis transform relating the two corresponding
enumerator formalisms.
This reflects a genuine structural difference: the operators \(\proj_\xi\)
span a commutative \(^\ast\)-algebra, whereas the operators \(\twirl_\xi\) need
not commute, and their span need not be closed under multiplication.
The MacWilliams-type duality needed for programming bounds appears only after
additional symmetry is imposed, as will be done in the subsequent sections.

\subsection{Sesquilinear Forms}

Each superoperator \(\op\in\L(\L(V))\) induces a sesquilinear form on \(\L(V)\)
via
\[
    (X_1,X_2)\mapsto \ip{X_1,\op(X_2)}.
    \numberthis
\]
In particular, associated with the projector \(\proj_\xi\) and the twirling map
\(\twirl_\xi\) we define
\begin{align}
    A_\xi(X_1,X_2) &:= \ip{X_1,\proj_\xi(X_2)}, \\
    B_\xi(X_1,X_2) &:= \ip{X_1,\twirl_\xi(X_2)}.
\end{align}
Using the expansion
\[
    \proj_\xi(X)=\sum_{E\in\B_\xi}\ip{E,X}\,E,
    \numberthis
\]
these forms may be written explicitly as
\begin{align}
A_\xi(X_1,X_2)
&=
\sum_{E\in\B_\xi}
\Tr(X_1^\dagger E)\Tr(X_2E^\dagger),
\label{eq:a-def}\\
B_\xi(X_1,X_2)
&=
\sum_{E\in\B_\xi}
\Tr(X_1^\dagger E^\dagger X_2E).
\label{eq:b-def}
\end{align}
Restricting to the diagonal \(X_1=X_2=X\) yields the associated quadratic forms
\begin{align}
    A_\xi(X,X) &= \|\proj_\xi(X)\|_2^2, \\
    B_\xi(X,X) &= \sum_{E\in\B_\xi}\Tr(X^\dagger E^\dagger X E).
\end{align}
We refer to the collections
\[
    \{A_\xi(X,X)\}_{\xi\in\mathcal E},
    \qquad
    \{B_\xi(X,X)\}_{\xi\in\mathcal E}
\]
as the intrinsic quadratic enumerators associated with the decomposition
\eqref{eq:operator-decomposition}.

\subsection{Normalization Identities}

The intrinsic quadratic enumerators satisfy the following basic positivity and
normalization relations.

\begin{lemma}
\label{lem:basic-identities}
For all \(X\in \L(V)\),
\begin{enumerate}[(1)]
\item \(A_\xi(X,X)\ge 0\);
\item \(\sum_\xi A_\xi(X,X)=\Tr(X^\dagger X)\);
\item \(\sum_\xi B_\xi(X,X)=|\Tr(X)|^2\);
\item If \(\E_0=\mathrm{span}\{I\}\), then
\[
    A_0(X,X)=\frac{|\Tr(X)|^2}{N},
    \qquad
    B_0(X,X)=\frac{\Tr(X^\dagger X)}{N}.
    \numberthis
\]
\end{enumerate}
\end{lemma}

\begin{proof}
For (1), by definition,
\[
    A_\xi(X,X)=\|\proj_\xi(X)\|_2^2 \ge 0.
    \numberthis
\]
For (2), since the decomposition
\[
    \L(V)=\bigoplus_{\xi\in\mathcal E}\E_\xi
    \numberthis
\]
is orthogonal, Parseval's identity gives
\[
    \sum_\xi A_\xi(X,X)
    =
    \sum_\xi \|\proj_\xi(X)\|_2^2
    =
    \|X\|_2^2
    =
    \Tr(X^\dagger X).
    \numberthis
\]
For (3), summing \eqref{eq:b-def} over \(\xi\) yields
\[
    \sum_\xi B_\xi(X,X)
    =
    \sum_{E\in \B}\Tr(X^\dagger E^\dagger X E),
    \numberthis
\]
where \(\B\) is the Hilbert--Schmidt orthonormal basis obtained by combining
the bases \(\B_\xi\).
Using the completeness identity,
\[
    \sum_{E\in\B} E^\dagger X E=\Tr(X)\,I,
    \numberthis
\]
we obtain
\begin{align}
    \sum_\xi B_\xi(X,X)
    &=
    \Tr\!\left(X^\dagger \sum_{E\in\B} E^\dagger X E\right) \\
    &=
    \Tr(X^\dagger \Tr(X)I) \\
    &=
    \Tr(X^\dagger)\Tr(X) \\
    &=
    |\Tr(X)|^2.
\end{align}
For (4), if \(\E_0=\mathrm{span}\{I\}\), then an orthonormal basis of
\(\E_0\) is
\[
    \B_0=\left\{\frac{1}{\sqrt{N}}I\right\}.
    \numberthis
\]
Substituting this into \eqref{eq:a-def} gives
\[
    A_0(X,X)
    =
    \Tr\!\left(X^\dagger \frac{I}{\sqrt{N}}\right)
    \Tr\!\left(X \frac{I}{\sqrt{N}}\right)
    =
    \frac{|\Tr(X)|^2}{N}.
    \numberthis
\]
Similarly, substituting into \eqref{eq:b-def} gives
\[
    B_0(X,X)
    =
    \Tr\!\left(X^\dagger \frac{I}{\sqrt{N}} X \frac{I}{\sqrt{N}}\right)
    =
    \frac{\Tr(X^\dagger X)}{N}.
    \numberthis
\]
This proves the claim.
\end{proof}

\subsection{A Fundamental Inequality}

The next lemma is the bridge from enumerators to error detection.
When \(X=P\) is a code projector, equality in the inequality below is exactly
the Knill--Laflamme detection condition for the sector \(\E_\xi\).

\begin{lemma}
\label{lem:ai-bi-ineq}
If \(X\) is positive semidefinite, then
\[
    A_\xi(X,X)\le \rank(X)\,B_\xi(X,X).
    \numberthis
\]
Moreover, equality holds if and only if
\[
    \sqrt{X}\,E\,\sqrt{X}=c_E\,P_X
    \qquad\text{for all }E\in\B_\xi,
    \numberthis
\]
where \(P_X\) is the orthogonal projector onto the image of \(X\) and each
\(c_E\in\CC\) is a scalar.
\end{lemma}

\begin{proof}
Let \(P_X\) denote the orthogonal projector onto the image of \(X\).
Since \(X\) is positive semidefinite, we may write
\[
    X=\sqrt{X}\sqrt{X},
    \qquad
    \sqrt{X}=P_X\sqrt{X}=\sqrt{X}P_X.
    \numberthis
\]
Using \eqref{eq:a-def}, we have
\[
    A_\xi(X,X)
    =
    \sum_{E\in\B_\xi}
    \Tr(X^\dagger E)\Tr(XE^\dagger).
    \numberthis
\]
Because \(X\) is positive semidefinite, \(X^\dagger=X\), and hence
\[
    \Tr(XE^\dagger)=\overline{\Tr(XE)}.
    \numberthis
\]
Therefore
\[
    A_\xi(X,X)
    =
    \sum_{E\in\B_\xi} |\Tr(XE)|^2.
    \numberthis
\]
Since \(X=\sqrt{X}\sqrt{X}\) and the trace is cyclic,
\[
    \Tr(XE)=\Tr(\sqrt{X}E\sqrt{X}),
    \numberthis
\]
so
\begin{equation}
    A_\xi(X,X)
    =
    \sum_{E\in\B_\xi} |\Tr(\sqrt{X}E\sqrt{X})|^2.
    \label{eq:A-as-traces}
\end{equation}
Now fix \(E\in\B_\xi\) and set
\[
    Y_E:=\sqrt{X}E\sqrt{X}.
    \numberthis
\]
Then \(Y_E=P_XY_EP_X\), so \(Y_E\) is supported on the range of \(P_X\).
Applying the Hilbert--Schmidt Cauchy--Schwarz inequality to \(Y_E\) and
\(P_X\) gives
\[
    |\Tr(Y_E^\dagger P_X)|^2
    \le
    \Tr(Y_E^\dagger Y_E)\Tr(P_X^\dagger P_X).
    \numberthis
\]
Since \(P_XY_E=Y_E\), we have
\[
    \Tr(Y_E^\dagger P_X)=\Tr(Y_E^\dagger)
    =
    \overline{\Tr(Y_E)},
    \numberthis
\]
and
\[
    \Tr(P_X^\dagger P_X)=\Tr(P_X)=\rank(X).
    \numberthis
\]
Thus
\[
    |\Tr(Y_E)|^2
    \le
    \rank(X)\,\Tr(Y_E^\dagger Y_E).
    \numberthis
\]
Substituting back \(Y_E=\sqrt{X}E\sqrt{X}\), we obtain
\[
    |\Tr(\sqrt{X}E\sqrt{X})|^2
    \le
    \rank(X)\,\Tr(XE^\dagger X E).
    \numberthis
\]
Summing over \(E\in\B_\xi\) and using \eqref{eq:A-as-traces} together with the
definition of \(B_\xi(X,X)\) yields
\[
    A_\xi(X,X)\le \rank(X)\,B_\xi(X,X).
    \numberthis
\]
Finally, equality in Hilbert--Schmidt Cauchy--Schwarz holds if and only if
\(Y_E\) is proportional to \(P_X\), i.e.
\[
    \sqrt{X}E\sqrt{X}=c_E P_X
    \numberthis
\]
for some scalar \(c_E\in\CC\).
Since the inequality above is obtained by summing the nonnegative quantities
\[
    \rank(X)\,\Tr(XE^\dagger X E)
    -
    |\Tr(\sqrt{X}E\sqrt{X})|^2,
    \numberthis
\]
equality in the summed inequality holds if and only if equality holds for each
\(E\in\B_\xi\) individually.
This proves the claim.
\end{proof}

\subsection{Specialization to Code Projectors}

We now specialize the preceding identities to the case of a code projector.
Let \(P:V\to V\) denote the orthogonal projector onto a code
\(\C\subseteq V\), and let
\[
    K:=\dim\C.
\]

\begin{corollary}
\label{cor:code-identities}
For a code projector \(P\),
\begin{enumerate}[(1)]
\item \(A_\xi(P,P)\ge 0\);
\item \(\sum_\xi A_\xi(P,P)=K\);
\item \(\sum_\xi B_\xi(P,P)=K^2\);
\item If \(\E_0=\mathrm{span}\{I\}\), then
\[
    A_0(P,P)=\frac{K^2}{N},
    \qquad
    B_0(P,P)=\frac{K}{N};
    \numberthis
\]
\item For every \(\xi\in\mathcal E\), we have
\[
    A_\xi(P,P)\le K\,B_\xi(P,P).
    \numberthis
\]
Equality holds if and only if
\[
    PEP=c_E P
    \qquad
    \text{for all }E\in\B_\xi
    \numberthis
\]
for suitable scalars \(c_E\in\CC\), i.e.\ if and only if \(\C\) detects the
sector \(\E_\xi\).
\end{enumerate}
\end{corollary}

\begin{proof}
Since \(P\) is an orthogonal projector, we have
\[
    P^\dagger=P,
    \qquad
    P^2=P,
    \qquad
    \Tr(P)=\rank(P)=K.
    \numberthis
\]
Thus (1)--(4) are immediate from Lemma~\ref{lem:basic-identities} applied to
\(X=P\).
Similarly, (5) follows from Lemma~\ref{lem:ai-bi-ineq} applied to \(X=P\),
since \(\rank(P)=K\).
The equality condition in Lemma~\ref{lem:ai-bi-ineq} becomes
\[
    \sqrt{P}\,E\,\sqrt{P}=c_E P.
    \numberthis
\]
Because \(\sqrt{P}=P\), this is equivalent to
\[
    PEP=c_E P
    \numberthis
\]
for all \(E\in\B_\xi\), which is exactly the Knill--Laflamme detection
condition for the sector \(\E_\xi\).
\end{proof}

\subsection{Recovery of the Shor--Laflamme enumerators}
\label{sec:sl-recovery}

The quantum Hamming-space case is already contained in Okada's
multiplicity-free quantum metric-space framework
\cite{okada2025quantumanalogdelsarteslinear}.
We spell it out here in the present projector/twirl language to make the
relation with the usual Shor--Laflamme enumerators explicit.

Let
\[
    V=(\CC^q)^{\otimes n},
    \qquad
    N:=\dim V=q^n.
\]
On one qudit, conjugation by \(U(q)\) gives the orthogonal decomposition
\[
    \L(\CC^q)
    =
    \CC I \oplus \L_0(\CC^q),
    \qquad
    \L_0(\CC^q):=\{X\in\L(\CC^q):\Tr(X)=0\}.
\]
Equivalently, as a \(U(q)\)-representation under conjugation,
\[
    \L(\CC^q)\cong \mathbf 1\oplus \Ad,
\]
where \(\Ad\) denotes the traceless adjoint sector.
Therefore
\[
    \L(V)
    =
    \L(\CC^q)^{\otimes n}
    =
    (\CC I\oplus \L_0(\CC^q))^{\otimes n}.
\]
For each subset \(S\subseteq[n]\), define
\[
    \E_S
    :=
    \bigotimes_{r=1}^n L_r(S),
    \qquad
    L_r(S):=
    \begin{cases}
        \L_0(\CC^q), & r\in S,\\
        \CC I, & r\notin S.
    \end{cases}
\]
Thus \(\E_S\) is the space of errors supported exactly on \(S\), and
\[
    \L(V)=\bigoplus_{S\subseteq[n]}\E_S.
\]
This is the support-refined decomposition associated with the local unitary
group \(U(q)^n\).

Now pass from the local group \(U(q)^n\) to the non-entangling wreath product
\[
    U(q)^n\rtimes S_n = U(q)\wr S_n,
\]
where \(S_n\) acts by permuting tensor factors.
The symmetric group permutes the support sectors \(\E_S\), and hence the
coarser \(S_n\)-orbit decomposition is indexed by the weight \(w=|S|\):
\[
    \E_w
    :=
    \bigoplus_{\substack{S\subseteq[n]\\ |S|=w}}\E_S,
    \qquad
    0\leq w\leq n.
\]
Thus \(\E_w\) is precisely the span of all \(n\)-qudit errors of Hamming weight
\(w\).
The orthogonal decomposition
\[
    \L(V)=\bigoplus_{w=0}^n \E_w
\]
is precisely the usual weight decomposition of the error space \(\L(V)\).

Let \(P\) be the projector onto a code \(\C\subseteq V\), and put
\(K:=\dim\C\).
Let \(\mathcal B_w\) be a Hilbert--Schmidt orthonormal basis of \(\E_w\).
The intrinsic enumerators of Section~\ref{sec:intrinsic-forms} are
\[
    A_w(P,P)
    =
    \sum_{E\in\mathcal B_w}|\Tr(PE)|^2,
    \qquad
    B_w(P,P)
    =
    \sum_{E\in\mathcal B_w}\Tr(PE^\dagger P E).
\]
We now compare these with the standard Shor--Laflamme enumerators.

Choose an unnormalized \(n\)-qudit Pauli/Weyl error basis $ \mathcal P_n $.
Then
\[
    \left\{N^{-1/2}E\right\}_{ E \in \mathcal P_n , \operatorname{wt}(E)=w}
\]
is a Hilbert--Schmidt orthonormal basis of \(\E_w\).
Hence
\begin{align}
    A_w(P,P)
    &=
    \frac{1}{N}
    \sum_{ E \in \mathcal P_n , \operatorname{wt}(E)=w}
    |\Tr(PE)|^2,
    \label{eq:intrinsic-Aw-SL}\\
    B_w(P,P)
    &=
    \frac{1}{N}
    \sum_{ E \in \mathcal P_n , \operatorname{wt}(E)=w}
    \Tr(PE^\dagger P E).
    \label{eq:intrinsic-Bw-SL}
\end{align}
With the usual coding-theoretic normalization,
\[
    A_w^{\mathrm{SL}}(P)
    :=
    \frac{1}{K^2}
    \sum_{ E \in \mathcal P_n , \operatorname{wt}(E)=w}
    |\Tr(P E)|^2,
\]
and
\[
    B_w^{\mathrm{SL}}(P)
    :=
    \frac{1}{K}
    \sum_{ E \in \mathcal P_n , \operatorname{wt}(E)=w}
    \Tr(P E^\dagger P E),
\]
we obtain
\[
    A_w^{\mathrm{SL}}(P)
    =
    \frac{N}{K^2}A_w(P,P),
    \qquad
    B_w^{\mathrm{SL}}(P)
    =
    \frac{N}{K}B_w(P,P).
\]
Equivalently, if one defines the normalized intrinsic enumerators by
\[
    \widetilde A_w:=\frac{N}{K^2}A_w(P,P),
    \qquad
    \widetilde B_w:=\frac{N}{K}B_w(P,P),
\]
then
\[
    \widetilde A_w=A_w^{\mathrm{SL}}(P),
    \qquad
    \widetilde B_w=B_w^{\mathrm{SL}}(P).
\]

Thus the Shor--Laflamme quantum weight enumerators are exactly the normalized
intrinsic enumerators for the non-entangling wreath-product action on
\((\CC^q)^{\otimes n}\).
The local group \(U(q)^n\) gives the finer support enumerators, while the
wreath product \(U(q)\wr S_n\) merges supports of the same cardinality and
gives the ordinary Hamming-weight enumerators.

Finally, the fundamental inequality of
Corollary~\ref{cor:code-identities} becomes
\[
    A_w^{\mathrm{SL}}(P)\leq B_w^{\mathrm{SL}}(P).
\]
Equality holds if and only if
\[
    PEP=c_E P
    \qquad
    \text{for every }E\in\E_w,
\]
or equivalently if and only if the code detects all errors of weight \(w\).
Therefore, for an ordinary distance-\(d\) quantum code, the intrinsic
sector-detection equalities
\[
    \widetilde A_w=\widetilde B_w,
    \qquad
    0\leq w<d,
\]
are precisely the usual Shor--Laflamme enumerator equalities.

\section{Equivariant MacWilliams Theory}
\label{sec:mult}

Section~II considered an arbitrary orthogonal decomposition of $\L(V)$ into
intrinsic error sectors and developed the associated quadratic enumerators,
normalization identities, and Knill--Laflamme inequalities. At that level of
generality, however, there is no canonical relation between the projector and
twirl enumerators. We now impose additional symmetry by assuming that $V$
carries a unitary representation of a group $G$ and taking the error sectors
to be the isotypic components of the induced conjugation representation on
\[
    \L(V)\cong V^*\otimes V.
\]
The symmetry forces both the projector and twirl constructions to lie in the
same equivariant intertwiner algebra and thereby produces a MacWilliams-type
duality between the corresponding enumerators.

In general, an irreducible representation may occur with nontrivial
multiplicity in $\L(V)$. The resulting intertwiner algebra
\[
    \Hom_G(\L(V),\L(V))
\]
is then noncommutative, and scalar data attached only to individual
irreducible copies do not capture its full structure. Instead, the
enumerators are naturally matrix-valued on the multiplicity spaces, the
MacWilliams transform is a unitary change of basis on the full intertwiner
algebra, and the corresponding positivity and Knill--Laflamme conditions are
matrix inequalities. Consequently, the natural optimization problem is a
semidefinite program rather than a linear program
\cite{BachocInvariantSDP}. We develop this general matrix-valued theory in the
remainder of this section.

A special case occurs when the conjugation representation is
multiplicity--free. Then every multiplicity space is one-dimensional, the
intertwiner algebra becomes commutative, the matrix-valued enumerators reduce
to scalars, and the semidefinite feasibility problem collapses to a linear
program. This multiplicity--free theory has been substantially developed
previously in the language of quantum metric spaces: the $\SU 2$ case goes
back to Bumgardner~\cite{bumgardner2012codeswastmetricspacestheory}, while
Okada developed the corresponding linear-programming framework for broader
classes of sufficiently symmetric quantum metric spaces
\cite{okada2025quantumanalogdelsarteslinear}. The multiplicity--free
specialization is therefore not a principal novelty of the present work. However for completeness, and to express it in the same projector/twirl notation as
the general theory, we develop this specialization in
Appendix~\ref{sec:mf-macwilliams}; symmetric-power and permutation-invariant
examples are collected in Appendix~\ref{sec:mf-examples}.

\subsection{Multiplicity spaces and block structure}

Let $V$ be a finite-dimensional unitary representation of a group $G$.
The operator space
\[
\L(V)\cong V^*\otimes V \numberthis
\]
is a unitary $G$-representation under conjugation, with canonical isotypic
decomposition
\[
\L(V)=\bigoplus_{\xi\in\widehat G}\E_\xi, \numberthis
\]
where \(\E_\xi\) is the \(\xi\)-isotypic component.
Let \(m_\xi\) denote the multiplicity of \(\xi\) in \(\L(V)\), and let
\(d_\xi:=\dim(\xi)\).

Choosing an orthogonal decomposition
\[
\E_\xi=\bigoplus_{\alpha=1}^{m_\xi}\E_{\xi\alpha},
\qquad
\E_{\xi\alpha}\cong \xi, \numberthis
\]
amounts to choosing an orthonormal basis of the multiplicity space attached to
\(\xi\).
Unlike the isotypic decomposition itself, this further splitting is not
canonical.
Equivalently, after such a choice one may identify
\[
\E_\xi\cong \xi\otimes \CC^{m_\xi}. \numberthis
\]

Schur's lemma then gives
\[
\Hom_G(\E_\xi,\E_\xi)\cong \text{End}(\CC^{m_\xi})\cong \CC^{m_\xi\times m_\xi}, \numberthis
\]
and hence
\[
\Hom_G(\L(V),\L(V))
\cong
\bigoplus_{\xi\in\widehat G}\CC^{m_\xi\times m_\xi}. \numberthis
\]
In particular,
\[
\dim \Hom_G(\L(V),\L(V))=\sum_\xi m_\xi^2, \numberthis
\]
which is strictly larger than \(\sum_\xi m_\xi\) whenever some \(m_\xi>1\).

Thus scalar data indexed only by \((\xi,\alpha)\) cannot capture the full
equivariant structure: the diagonal directions account for only \(m_\xi\)
degrees of freedom in the \(\xi\)-block, whereas the full intertwiner algebra on
that block has dimension \(m_\xi^2\).
The missing directions are the off-diagonal intertwiners between distinct copies
\(\E_{\xi\alpha}\) and \(\E_{\xi\beta}\).
Accordingly, the equivariant enumerator data are naturally matrix-valued on
each multiplicity space.

The role of the off-diagonal entries can be seen directly. Fix a sector
\(\xi\), and for \(c=(c_1,\ldots,c_{m_\xi})\in\CC^{m_\xi}\) form the coherent
error directions
\[
    U_i(c):=\sum_{\alpha=1}^{m_\xi}c_\alpha E_i^\alpha,
    \qquad 1\le i\le d_\xi .
\]
The scalar projector and twirl quantities associated with this coherent
direction are precisely \(c^\dagger A_\xi c\) and \(c^\dagger B_\xi c\). Hence
requiring the rank-\(K\) inequality for every coherent combination \(c\) is
equivalent to a Loewner inequality on the whole block,
\[
    c^\dagger A_\xi c\le K\,c^\dagger B_\xi c\ \ \text{for all }c
    \quad\Longleftrightarrow\quad
    A_\xi\preceq K B_\xi ,
\]
and detection of the full isotypic sector is equivalent to \(A_\xi=K B_\xi\).
Retaining only the diagonal entries would test the basis directions
\(c=e_\alpha\) but discard the interference terms that control coherent
superpositions of equivalent irreducible copies. This is the elementary reason
that multiplicity forces matrix-valued enumerators and semidefinite, rather
than linear, constraints.

The next two subsections construct two orthonormal bases of
\(\Hom_G(\L(V),\L(V))\): a projector matrix-unit basis built from generalized
projectors, and a twirl-derived orthogonal basis built from generalized twirls.
The matrix-valued MacWilliams transform will arise as the change of basis
between these two families.

\subsection{Projector matrix units}

For each \(\xi\), let
\[
\Pi_{\xi\alpha}:\L(V)\to\L(V) \numberthis
\]
denote the orthogonal projector onto \(\E_{\xi\alpha}\).
Choose isometric \(G\)-intertwiners
\[
\phi_{\alpha\beta}^{\xi}:\E_{\xi\alpha}\to\E_{\xi\beta} \numberthis
\]
satisfying
\[
\phi_{\alpha\alpha}^{\xi}=1,
\qquad
\phi_{\beta\alpha}^{\xi}=(\phi_{\alpha\beta}^{\xi})^\dagger,
\qquad
\phi_{\beta\gamma}^{\xi}\phi_{\alpha\beta}^{\xi}=\phi_{\alpha\gamma}^{\xi}. \numberthis
\]
For example, one may choose arbitrary isometries
\(\phi_{1\alpha}^{\xi}:\E_{\xi1}\to\E_{\xi\alpha}\) and then define
\[
\phi_{\alpha\beta}^{\xi}:=\phi_{1\beta}^{\xi}(\phi_{1\alpha}^{\xi})^\dagger. \numberthis
\]

We define the associated \emph{projector matrix units} by
\begin{equation}\label{eq:projector-matrix-unit}
\Pi_{\xi\alpha\beta}:=\phi_{\alpha\beta}^{\xi}\,\Pi_{\xi\alpha}.
\end{equation}
For \(\alpha=\beta\), this recovers the orthogonal projector onto
\(\E_{\xi\alpha}\). For \(\alpha\neq\beta\), it is the corresponding
off-diagonal matrix unit.

Choose an orthonormal basis
\[
\B_{\xi\alpha}=\{E_i^\alpha\}_{i=1}^{d_\xi} \numberthis
\]
of \(\E_{\xi\alpha}\), and define the corresponding basis of \(\E_{\xi\beta}\) by
\[
E_i^\beta:=\phi_{\alpha\beta}^{\xi}(E_i^\alpha),
\qquad
1\le i\le d_\xi. \numberthis
\]
Then
\begin{equation}\label{eq:projector-matrix-unit-formula}
\Pi_{\xi\alpha\beta}(X)
=
\sum_{i=1}^{d_\xi}\ip{E_i^\alpha,X}\,E_i^\beta,
\qquad
X\in\L(V). \numberthis
\end{equation}

Each \(\Pi_{\xi\alpha\beta}\) is \(G\)-equivariant, and the matrix-unit
relations hold:
\begin{equation}\label{eq:projector-matrix-unit-rel}
\Pi_{\xi\alpha\beta}\Pi_{\rho\mu\nu}
=
\delta_{\xi\rho}\,\delta_{\alpha\nu}\,\Pi_{\xi\mu\beta}.  \numberthis
\end{equation}
Moreover, with respect to the Hilbert--Schmidt inner product on
\(\L(\L(V))\)
\begin{equation}\label{eq:projector-matrix-unit-orth}
\lsuper \Pi_{\xi\alpha\beta},\Pi_{\rho\mu\nu}\rsuper
=
\Tr(\Pi_{\xi\alpha\beta}^\dagger\Pi_{\rho\mu\nu})
=
\delta_{\xi\rho}\,\delta_{\alpha\mu}\,\delta_{\nu\beta}\,d_\xi. \numberthis
\end{equation}
Hence
\[
\left\{\frac{1}{\sqrt{d_\xi}}\Pi_{\xi\alpha\beta}\right\}_{\xi,\alpha,\beta} \numberthis
\]
is an orthonormal basis of \(\Hom_G(\L(V),\L(V))\).

\subsection{The twirl-derived orthogonal basis}

Fix orthonormal bases \(\B_{\xi\alpha}\) of each \(\E_{\xi\alpha}\).
For each \(\xi,\alpha,\beta\), define the \emph{generalized twirls} by
\begin{equation}\label{eq:twirl-matrix-unit}
\Twirl_{\xi\alpha\beta}(M)
:=
\sum_{E\in\B_{\xi\alpha}} E^\dagger\,M\,\phi_{\alpha\beta}^{\xi}(E),
\qquad
M\in\L(V). \numberthis
\end{equation}
Writing \(\B_{\xi\alpha}=\{E_i^\alpha\}_{i=1}^{d_\xi}\) and
\(E_i^\beta:=\phi_{\alpha\beta}^{\xi}(E_i^\alpha)\), this becomes
\begin{equation}\label{eq:twirl-matrix-unit-ordered}
\Twirl_{\xi\alpha\beta}(M)
=
\sum_{i=1}^{d_\xi} (E_i^\alpha)^\dagger\,M\,E_i^\beta.
\end{equation}
This definition is independent of the choice of orthonormal basis
\(\B_{\xi\alpha}\), and each \(\Twirl_{\xi\alpha\beta}\) is \(G\)-equivariant.

To compute their inner products, let \(\B\) be any orthonormal basis of
\(\L(V)\). Then
\[
\lsuper \Twirl_{\xi\alpha\beta},\Twirl_{\rho\mu\nu}\rsuper
=
\sum_{M\in\B}
\ip{\Twirl_{\xi\alpha\beta}(M),\Twirl_{\rho\mu\nu}(M)}. \numberthis
\]
Substituting \eqref{eq:twirl-matrix-unit-ordered} gives
\[
\sum_{M\in\B}\sum_{i,j}
\Tr\!\left(
(E_i^\beta)^\dagger M^\dagger E_i^\alpha
(E_j^\mu)^\dagger M E_j^\nu
\right). \numberthis
\]
Using the completeness identity from Section~\ref{sec:intrinsic-forms}, the sum
over \(M\) collapses to
\[
\sum_{i,j}
\Tr\!\left((E_i^\beta)^\dagger E_j^\nu\right)
\Tr\!\left(E_i^\alpha (E_j^\mu)^\dagger\right). \numberthis
\]
Since the families \(\{E_i^\alpha\}\) and \(\{E_j^\mu\}\) are orthonormal and
mutually orthogonal unless \((\xi,\alpha)=(\rho,\mu)\), this becomes
\[
\lsuper\Twirl_{\xi\alpha\beta},\Twirl_{\rho\mu\nu}\rsuper
=
\delta_{\xi\rho}\,\delta_{\alpha\mu}\,\delta_{\beta\nu}\,d_\xi. \numberthis
\]
Hence
\[
\left\{\frac{1}{\sqrt{d_\xi}}\Twirl_{\xi\alpha\beta}\right\}_{\xi,\alpha,\beta} \numberthis
\]
is a second orthonormal basis of \(\Hom_G(\L(V),\L(V))\), which we call the
\emph{twirl-derived orthogonal basis}.

\subsection{The matrix-valued MacWilliams transform}

The projector matrix-unit basis and the twirl-derived orthogonal basis
constructed above are two orthonormal bases of the same Hilbert space
\[
\Hom_G(\L(V),\L(V)) \numberthis
\]
with respect to the Hilbert--Schmidt inner product.
Consequently the change of basis between them is unitary.
More precisely, define the matrix \(U\) with rows indexed by the output triple
\((\rho,\mu,\nu)\) and columns by the source triple \((\xi,\alpha,\beta)\) by
\begin{equation}\label{eq:U-mult-def}
[U]_{(\rho,\mu,\nu),(\xi,\alpha,\beta)}
:=
\Big\langle
\frac{1}{\sqrt{d_{\xi}}}\Pi_{\xi\alpha\beta},
\frac{1}{\sqrt{d_{\rho}}}\Twirl_{\rho\mu\nu}
\Big\rangle . \numberthis
\end{equation}

\begin{proposition}[Block MacWilliams transform]
\label{prop:block-macwilliams}
The matrix \(U\) defined in \eqref{eq:U-mult-def} is unitary.
It gives the change of basis between the projector matrix-unit basis
\(\{\Pi_{\xi\alpha\beta}\}\) and the twirl-derived orthogonal basis
\(\{\Twirl_{\xi\alpha\beta}\}\) of \(\Hom_G(\L(V),\L(V))\).
\end{proposition}

\begin{proof}
Both families
\[
\left\{\frac{1}{\sqrt{d_\xi}}\Pi_{\xi\alpha\beta}\right\}_{\xi,\alpha,\beta},
\qquad
\left\{\frac{1}{\sqrt{d_\xi}}\Twirl_{\xi\alpha\beta}\right\}_{\xi,\alpha,\beta} \numberthis
\]
are orthonormal bases of the same finite--dimensional Hilbert space
\(\Hom_G(\L(V),\L(V))\).
Therefore the matrix of inner products between them is unitary.
\end{proof}

Expanding \eqref{eq:U-mult-def} using the explicit formulas for the matrix units
yields a trace expression.  Writing
\(\B_{\xi\alpha}=\{E_i^\alpha\}_{i=1}^{d_\xi}\) and
\(E_i^\beta=\phi_{\alpha\beta}^{\xi}(E_i^\alpha)\), we obtain
\begin{equation}\label{eq:U-mult-trace}
[U]_{(\rho,\mu,\nu),(\xi,\alpha,\beta)}
=
\frac{1}{\sqrt{d_\xi d_\rho}}
\sum_{i=1}^{d_\xi}\sum_{j=1}^{d_\rho}
\Tr\!\left(
(E_i^\beta)^\dagger (E_j^\mu)^\dagger
E_i^\alpha E_j^\nu
\right). \numberthis
\end{equation}

To pass from the unitary basis change \(U\) to the transform acting on the
\emph{unnormalized} enumerators, we must undo the Plancherel normalization on
the projector side. Write \(a=(\xi,\alpha,\beta)\) for a source index and
\(b=(\rho,\mu,\nu)\) for an output index, with \(d_a=d_\xi\) and
\(d_b=d_\rho\), and let
\begin{equation}\label{eq:R-mult-def}
[R]_{ba}:=\langle \Pi_{a},\Twirl_{b}\rangle_{\mathrm{HS}}
=\sqrt{d_a d_b}\,[U]_{ba} \numberthis
\end{equation}
be the raw overlap. Expanding a twirl in the (non-normalized) projector basis
divides by \(\lVert\Pi_a\rVert_{\mathrm{HS}}^2=d_a\); the resulting coefficient
matrix is therefore
\begin{equation}\label{eq:M-mult-def}
[M]_{ba}
:=
\sqrt{\frac{d_b}{d_a}}\,[U]_{ba}
=
\frac{[R]_{ba}}{d_a},
\qquad
D_{aa}:=d_a. \numberthis
\end{equation}
Equivalently,
\[
M=D^{1/2}UD^{-1/2},
\]
with the output index as the row of \(M\). The raw overlap \([R]\) is \emph{not} the
enumerator coefficient matrix; the factor \(1/d_a\) is required. The matrix
\(M\) is generally not unitary, but it obeys the weighted-orthogonality
relation
\begin{equation}\label{eq:M-mult-weighted-orth}
\boxed{\,M\,D\,M^\dagger=D\,,} \numberthis
\end{equation}
since
\[
[MDM^\dagger]_{bb'}
=\sum_a \sqrt{\tfrac{d_b}{d_a}}\,U_{ba}\;d_a\;
        \sqrt{\tfrac{d_{b'}}{d_a}}\,\overline{U_{b'a}}
=\sqrt{d_b d_{b'}}\sum_a U_{ba}\overline{U_{b'a}}
=\sqrt{d_b d_{b'}}\,[UU^\dagger]_{bb'}
=\sqrt{d_b d_{b'}}\,\delta_{bb'}
=D_{bb'},
\]
using unitarity of \(U\).

\begin{lemma}[Trivial-row audit]
\label{lem:trivial-row-audit}
If \(V\) is irreducible of dimension \(N\), the trivial sector is
\(\E_{\mathbf1}=\CC I\) and \(\Twirl_{\mathbf1}=\tfrac1N\,\mathrm{id}_{\L(V)}\),
so
\[
B_{\mathbf1}(X)=\frac1N\sum_{\xi,\alpha}[A_\xi(X)]_{\alpha\alpha}.
\]
Equivalently, the trivial row of \(M\) equals \(1/N\) on every diagonal
\(A\)-coordinate and \(0\) on every off-diagonal coordinate. This identity is a
convenient check on any explicitly computed coefficient matrix \(M\).
\end{lemma}

Different coherent choices of the intertwiners
\(\phi_{\alpha\beta}^{\xi}\) correspond to changing the orthonormal basis of the
multiplicity space \(\CC^{m_\xi}\) attached to each \(\xi\).
Such a change conjugates the enumerators \(A_\xi\) and \(B_\xi\) by a unitary
acting on \(\CC^{m_\xi}\), and hence does not affect any positivity or
feasibility conditions in the semidefinite programs derived below.

\subsection{Matrix-valued intrinsic enumerators}

We now associate intrinsic enumerator data to the multiplicity spaces
introduced above. For each \(\xi,\alpha,\beta\), associate to the projector matrix unit
\(\Pi_{\xi\alpha\beta}\) the sesquilinear form
\[
a_{\xi\alpha\beta}(X_1,X_2)
:=
\langle X_1,\Pi_{\xi\alpha\beta}(X_2)\rangle , \numberthis
\]
and to the generalized twirl \(\Twirl_{\xi\alpha\beta}\) the form
\[
b_{\xi\alpha\beta}(X_1,X_2)
:=
\langle X_1,\Twirl_{\xi\alpha\beta}(X_2)\rangle .\numberthis
\]

Using the explicit formulas for the matrix units, these forms admit the
computational expressions
\begin{align}
a_{\xi\alpha\beta}(X_1,X_2)
&=
\sum_{i=1}^{d_\xi}
\Tr(X_1^\dagger E_i^\beta)\,
\Tr(X_2 (E_i^\alpha)^\dagger),
\label{eq:a-formula}
\\
b_{\xi\alpha\beta}(X_1,X_2)
&=
\sum_{i=1}^{d_\xi}
\Tr\!\big(X_1^\dagger (E_i^\alpha)^\dagger X_2 E_i^\beta\big). \numberthis
\label{eq:b-formula}
\end{align}

For a fixed operator \(X\in\L(V)\) and fixed irreducible label \(\xi\), define
the \(m_\xi\times m_\xi\) matrices
\begin{align}
[A_\xi(X)]_{\alpha\beta} &:= a_{\xi\alpha\beta}(X,X),
\label{eq:A-xi-def}
\\
[B_\xi(X)]_{\alpha\beta} &:= b_{\xi\alpha\beta}(X,X).
\label{eq:B-xi-def} \numberthis
\end{align}

We call \(A_\xi(X)\) and \(B_\xi(X)\) the \emph{matrix-valued intrinsic
enumerators} of \(X\) in the \(\xi\)-isotypic sector. The matrix-valued enumerators are related by the block MacWilliams transform.

\begin{proposition}[Matrix MacWilliams identity]
\label{prop:matrix-macwilliams}
Let \(X\in\L(V)\).
Then the matrices \(A_\xi(X)\) and \(B_\xi(X)\) satisfy
\begin{equation}
[B_\rho(X)]_{\mu\nu}
=
\sum_{\xi}\sum_{\alpha,\beta}
[M]_{(\rho,\mu,\nu),(\xi,\alpha,\beta)}
\,[A_\xi(X)]_{\alpha\beta}. \numberthis
\label{eq:block-macwilliams}
\end{equation}
Equivalently, if the matrices \(A_\xi(X)\) and \(B_\xi(X)\) are vectorized and
stacked over all \(\xi\), then
\[
\mathrm{vec}(B)=M\,\mathrm{vec}(A), \numberthis
\]
where \(M\) is the block MacWilliams matrix defined in
\eqref{eq:M-mult-def}.
\end{proposition}

\begin{proof}
The identity follows directly from the change of basis between the projector
matrix-unit basis and the twirl-derived orthogonal basis.
Indeed, expanding \(\Twirl_{\rho\mu\nu}\) in the projector basis using the
matrix \(M\) and evaluating the inner products defining
\(a_{\xi\alpha\beta}\) and \(b_{\rho\mu\nu}\) yields
\eqref{eq:block-macwilliams}.
\end{proof}

\subsection{Block positivity and normalization}

The matrix-valued intrinsic enumerators satisfy natural positivity and
normalization properties.
For Hermitian matrices we use the Loewner order
\[
A \preceq B, \numberthis
\]
meaning that \(B-A\) is positive semidefinite.

\begin{lemma}[Block positivity]\label{lem:block-psd}
Let \(X\in\L(V)\).
For every irreducible label \(\xi\), the matrix \(A_\xi(X)\) is Hermitian
positive semidefinite.
If moreover \(X\succeq 0\), then \(B_\xi(X)\) is also Hermitian positive
semidefinite.

In addition, the diagonal traces satisfy the normalization identities
\begin{align}
\sum_{\xi}\Tr A_\xi(X) &= \Tr(X^\dagger X), \label{eq:A-parseval}\\
\sum_{\xi}\Tr B_\xi(X) &= |\Tr(X)|^2. \numberthis \label{eq:B-parseval}
\end{align}
\end{lemma}

\begin{proof}
Let \(c\in\mathbb{C}^{m_\xi}\). Using \eqref{eq:a-formula}, we compute
\begin{align*}
\sum_{\alpha,\beta}\overline{c_\alpha}c_\beta [A_\xi(X)]_{\alpha\beta}
&=
\sum_{\alpha,\beta}\overline{c_\alpha}c_\beta
\sum_{i=1}^{d_\xi}
\Tr(X^\dagger E_i^\beta)\,
\Tr(X (E_i^\alpha)^\dagger)
\\
&=
\sum_{i=1}^{d_\xi}
\Big(
\sum_{\beta} c_\beta \Tr(X^\dagger E_i^\beta)
\Big)
\Big(
\sum_{\alpha} \overline{c_\alpha} \Tr(X (E_i^\alpha)^\dagger)
\Big). \numberthis
\end{align*}

Since
\[
\sum_{\alpha} \overline{c_\alpha} \Tr(X (E_i^\alpha)^\dagger)= \Big( \sum_{\alpha} c_\alpha \Tr(X^\dagger E_i^\alpha) \Big)^\dagger \numberthis
\]
the two factors are complex conjugates. Hence
\[
\sum_{\alpha,\beta}\overline{c_\alpha}c_\beta [A_\xi(X)]_{\alpha\beta}
=
\sum_{i=1}^{d_\xi}
\Big|
\Tr\!\Big(
X^\dagger \sum_{\beta=1}^{m_\xi} c_\beta E_i^\beta
\Big)
\Big|^2
\ge 0. \numberthis
\]
and we can conclude that \(A_\xi(X)\succeq 0\).

Similarly, using \eqref{eq:b-formula} and defining
\[
U_i:=\sum_{\beta=1}^{m_\xi} c_\beta E_i^\beta, \numberthis
\]
we obtain
\[
\sum_{\alpha,\beta}\overline{c_\alpha}c_\beta [B_\xi(X)]_{\alpha\beta}
=
\sum_{i=1}^{d_\xi}\Tr(X^\dagger U_i^\dagger X U_i). \numberthis
\]
If \(X\succeq 0\), then \(X=X^\dagger\), so
\begin{align*}
    \Tr(X^\dagger U_i^\dagger X U_i) &=
\Tr(XU_i^\dagger XU_i) \\
&=
\Tr\!\big((X^{1/2}U_iX^{1/2})^\dagger(X^{1/2}U_iX^{1/2})\big) \\
&=
\langle X^{1/2}U_iX^{1/2}, X^{1/2}U_iX^{1/2} \rangle  \numberthis
\end{align*}
Therefore \(B_\xi(X)\succeq 0\).

For the normalization identities, note first that the diagonal projector matrix
units are exactly the orthogonal projections onto the summands
\(\E_{\xi\alpha}\), so
\[
\sum_{\xi,\alpha}\Pi_{\xi\alpha\alpha}=I_{\L(V)}. \numberthis
\]
Hence
\[
\sum_\xi \Tr A_\xi(X)
=
\sum_{\xi,\alpha}\langle X,\Pi_{\xi\alpha\alpha}(X)\rangle
=
\langle X,X\rangle
=
\Tr(X^\dagger X). \numberthis
\]

Similarly, the union of the orthonormal bases \(\B_{\xi\alpha}\) over all
\(\xi,\alpha\) is an orthonormal basis of \(\L(V)\). Therefore the completeness
identity from Section~\ref{sec:intrinsic-forms} gives
\[
\sum_{\xi,\alpha}\Twirl_{\xi\alpha\alpha}(X)=\Tr(X)\,I, \numberthis
\]
and so
\begin{align*}
    \sum_\xi \Tr B_\xi(X)
&=
\sum_{\xi,\alpha}\langle X,\Twirl_{\xi\alpha\alpha}(X)\rangle \\
&=
\langle X,\Tr(X)I\rangle \\
&=
|\Tr(X)|^2 \numberthis
\end{align*}
\end{proof}

\subsection{A matrix Knill--Laflamme inequality}

The fundamental inequality \(A_i(X,X)\le \rank(X)\,B_i(X,X)\) of Section~\ref{sec:intrinsic-forms} extends
naturally to the multiplicity spaces. For a fixed irreducible sector \(\xi\),
the inequality must hold simultaneously for every coherent combination of its
equivalent irreducible copies. As shown below, this family of scalar
inequalities is equivalent to a single Loewner-order inequality for the
matrices \(A_\xi(X)\) and \(B_\xi(X)\).

\begin{lemma}[Matrix Knill--Laflamme inequality]\label{lem:matrix-KL}
Let \(X\in\L(V)\) be positive semidefinite, and let $ r:=\rank(X) $. Then for every irreducible label \(\xi\),
\begin{equation}\label{eq:matrix-KL}
A_\xi(X)\preceq r\,B_\xi(X). \numberthis
\end{equation}
In particular, if \(X=P\) is the orthogonal projector onto a
\(K\)-dimensional code subspace \(\C\subseteq V\), then
\begin{equation}\label{eq:matrix-KL-code}
A_\xi(P)\preceq K\,B_\xi(P). \numberthis
\end{equation}
\end{lemma}

\begin{proof}
Fix \(\xi\), and let \(c=(c_\beta)_{\beta=1}^{m_\xi}\in\CC^{m_\xi}\).
To prove \eqref{eq:matrix-KL}, it suffices to show that
\[
c^\dagger A_\xi(X)c \le r\, c^\dagger B_\xi(X)c
\qquad\text{for all }c\in\CC^{m_\xi}. \numberthis
\]

Since \(X\succeq 0\), it is Hermitian and admits a positive square root.
Thus
\[
X=X^\dagger,
\qquad
X=\sqrt{X}\,\sqrt{X}. \numberthis
\]

Define
\[
U_i:=\sum_{\beta=1}^{m_\xi} c_\beta E_i^\beta,
\qquad 1\le i\le d_\xi. \numberthis
\]
Using \eqref{eq:a-formula} and \eqref{eq:b-formula}, we compute
\begin{align}
c^\dagger A_\xi(X)c
&=
\sum_{\alpha,\beta}\overline{c_\alpha}c_\beta [A_\xi(X)]_{\alpha\beta}
\notag\\
&=
\sum_{i=1}^{d_\xi}
\left|
\Tr\!\left(
X^\dagger \sum_{\beta=1}^{m_\xi} c_\beta E_i^\beta
\right)
\right|^2
\notag\\
&=
\sum_{i=1}^{d_\xi} |\Tr(XU_i)|^2
\notag\\
&=
\sum_{i=1}^{d_\xi} |\Tr(\sqrt{X}\,U_i\,\sqrt{X})|^2 , \numberthis
\label{eq:matrix-KL-A}
\end{align}
where in the last step we used \(X=\sqrt{X}\sqrt{X}\) and cyclicity of trace.

Similarly,
\begin{align}
c^\dagger B_\xi(X)c
&=
\sum_{\alpha,\beta}\overline{c_\alpha}c_\beta [B_\xi(X)]_{\alpha\beta}
\notag\\
&=
\sum_{i=1}^{d_\xi}\Tr(X^\dagger U_i^\dagger XU_i)
\notag\\
&=
\sum_{i=1}^{d_\xi}\Tr(XU_i^\dagger XU_i). \numberthis
\label{eq:matrix-KL-B}
\end{align}

Now set
\[
Y_i:=\sqrt{X}\,U_i\,\sqrt{X}. \numberthis
\]
Since \(\rank(Y_i)\le \rank(\sqrt{X})=\rank(X)=r\), the elementary inequality
\[
|\Tr(Y)|^2 \le \rank(Y)\,\|Y\|_2^2 \numberthis
\]
gives
\[
|\Tr(Y_i)|^2 \le r\,\|Y_i\|_2^2 . \numberthis
\]
Also,
\begin{align*}
\|Y_i\|_2^2
&=
\Tr(Y_i^\dagger Y_i) \\
&=
\Tr\!\big((\sqrt{X}U_i\sqrt{X})^\dagger(\sqrt{X}U_i\sqrt{X})\big) \\
&=
\Tr\!\big(\sqrt{X}\,U_i^\dagger XU_i\,\sqrt{X}\big) \\
&=
\Tr(XU_i^\dagger XU_i). \numberthis
\end{align*}
Therefore,
\[
|\Tr(\sqrt{X}\,U_i\,\sqrt{X})|^2
\le
r\,\Tr(XU_i^\dagger XU_i). \numberthis
\]
Summing over \(i\) and using \eqref{eq:matrix-KL-A}--\eqref{eq:matrix-KL-B}
yields
\[
c^\dagger A_\xi(X)c
\le
r\,c^\dagger B_\xi(X)c. \numberthis
\]
Since this holds for every \(c\in\CC^{m_\xi}\), we conclude that
\[
A_\xi(X)\preceq r\,B_\xi(X), \numberthis
\]
which is \eqref{eq:matrix-KL}.

If \(X=P\) is the projector onto a \(K\)-dimensional code subspace, then $ \rank(P)=K $ so \eqref{eq:matrix-KL-code} follows immediately from \eqref{eq:matrix-KL}.
\end{proof}

For a code projector \(P\), the inequality
\eqref{eq:matrix-KL-code} is the intrinsic Knill--Laflamme inequality on the
full multiplicity space. The following equality characterization identifies
precisely when the corresponding isotypic sector is detected and gives the
detection criterion used in the semidefinite program below.

\begin{theorem}[Matrix Knill--Laflamme equality]
\label{thm:matrix-KL-equality}
Let \(P\) be the orthogonal projector onto a \(K\)-dimensional code
\(\C\subseteq V\), and fix an irreducible label \(\xi\). Then
\(A_\xi(P)\preceq K\,B_\xi(P)\), and the following are equivalent:
\begin{enumerate}[(i)]
\item \(A_\xi(P)=K\,B_\xi(P)\);
\item \(PEP=c_E P\) for every \(E\in\E_\xi\), i.e.\ \(\C\) detects the isotypic
sector \(\E_\xi\).
\end{enumerate}
\end{theorem}

\begin{proof}
For \(c=(c_\beta)\in\CC^{m_\xi}\) put \(U_i(c):=\sum_{\beta}c_\beta E_i^\beta\).
By \eqref{eq:a-formula} and \eqref{eq:b-formula}, and using \(P^\dagger=P\),
\[
c^\dagger A_\xi(P)c=\sum_{i=1}^{d_\xi}\bigl|\Tr(P\,U_i(c))\bigr|^2,
\qquad
c^\dagger B_\xi(P)c=\sum_{i=1}^{d_\xi}\bigl\lVert P\,U_i(c)\,P\bigr\rVert_2^2 .
\]
Since \(\Tr(P\,U_i(c))=\Tr(P\,U_i(c)\,P)=\langle P,\,P\,U_i(c)\,P\rangle\) and
\(\lVert P\rVert_2^2=\Tr(P)=K\), the Cauchy--Schwarz inequality gives
\[
\bigl|\Tr(P\,U_i(c))\bigr|^2\le K\,\bigl\lVert P\,U_i(c)\,P\bigr\rVert_2^2,
\]
and summing over \(i\) yields \(c^\dagger A_\xi(P)c\le K\,c^\dagger B_\xi(P)c\)
for all \(c\), i.e.\ \(A_\xi(P)\preceq K B_\xi(P)\) (recovering
\eqref{eq:matrix-KL-code}). If equality \(A_\xi(P)=K\,B_\xi(P)\) holds, then for
every \(c\) the sum of the nonnegative Cauchy--Schwarz deficits vanishes, so
each deficit vanishes and \(P\,U_i(c)\,P\in\CC P\). Taking \(c=e_\beta\) gives
\(P E_i^\beta P\in\CC P\) for all \(i,\beta\); as the \(E_i^\beta\) span
\(\E_\xi\), we obtain \(PEP=c_E P\) for all \(E\in\E_\xi\). The converse is
immediate, since detection makes every Cauchy--Schwarz step tight.
\end{proof}

\subsection{Semidefinite Programming Bounds in the General Equivariant Case}
\label{subsec:enumerators-to-sdp}

We now show that the existence of an intrinsic code with prescribed detection
properties implies feasibility of a finite system of matrix equations and
semidefinite inequalities.
This gives the basic semidefinite programming relaxation in the general
equivariant setting.

Let $ M$ denote the intrinsic MacWilliams transform in the unnormalized block basis, so
that for every code projector \(P\) the corresponding matrix-valued projector
and twirl enumerators satisfy
\[
\mathrm{vec}(B)=M\,\mathrm{vec}(A). \numberthis
\]
Here \(\mathrm{vec}(A)\) denotes the vector obtained by stacking the entries of
the matrices \(A_\xi\) over all irreducible labels \(\xi\), and similarly for
\(\mathrm{vec}(B)\).

\begin{theorem}
\label{thm:intrinsic-sdp-feasibility}
Let \(V\) be a finite-dimensional unitary representation of a group \(G\), and
write
\[
\L(V)=\bigoplus_{\xi\in\Lambda}\E_\xi \numberthis
\]
for the isotypic decomposition of the conjugation representation, where
\(\Lambda\) is the set of irreducible labels occurring in \(\L(V)\).
Suppose there exists a code subspace \(\C\subseteq V\) of dimension \(K:=\dim\C\)
with orthogonal projector \(P\), and suppose that \(\C\) detects every error in
each sector \(\E_\xi\) for \(\xi\in\mathcal D\), where \(\mathcal D\subseteq\Lambda\)
is the prescribed set of \emph{detected} labels (as distinct from \(\Lambda\),
the set of all occurring labels).
Then there exists a list of Hermitian matrices indexed by \(\xi\in\Lambda\),
\[
A_\xi,\;B_\xi\in\CC^{m_\xi\times m_\xi} \numberthis
\]
satisfying the following equations and semidefinite inequalities:
\begin{align}
0  &\preceq A_\xi \qquad\text{for all }\xi\in\Lambda, \label{eq:sdp-thm-1}\\
0  &\preceq B_\xi \qquad\text{for all }\xi\in\Lambda, \label{eq:sdp-thm-2}\\
A_\xi &\preceq K\,B_\xi \qquad\text{for all }\xi\in\Lambda, \label{eq:sdp-thm-3}\\
\sum_{\xi\in\Lambda} \Tr A_\xi &= K, \label{eq:sdp-thm-4}\\
\sum_{\xi\in\Lambda} \Tr B_\xi &= K^2, \label{eq:sdp-thm-5}\\
\mathrm{vec}(B) &= M\,\mathrm{vec}(A), \label{eq:sdp-thm-6}\\
A_\xi &= K\,B_\xi
\qquad\text{for all }\xi\in\mathcal D. \numberthis \label{eq:sdp-thm-7}
\end{align}
In fact, one may take
\[
A_\xi=A_\xi(P),
\qquad
B_\xi=B_\xi(P). \numberthis
\]
If in addition \(V\) is irreducible of dimension \(N\), then
\(A_{\mathbf1}=K^2/N\) and \(B_{\mathbf1}=K/N\).
\end{theorem}

\begin{proof}
Set
\[
A_\xi:=A_\xi(P),
\qquad
B_\xi:=B_\xi(P). \numberthis
\]
Then \eqref{eq:sdp-thm-1} and \eqref{eq:sdp-thm-2} follow from
Lemma~\ref{lem:block-psd} (with \(X=P\succeq0\)), while \eqref{eq:sdp-thm-4} and
\eqref{eq:sdp-thm-5} are exactly the normalization identities
\eqref{eq:A-parseval} and \eqref{eq:B-parseval} applied to \(X=P\).

Equation \eqref{eq:sdp-thm-3} is the matrix Knill--Laflamme inequality
\eqref{eq:matrix-KL-code}.
Equation \eqref{eq:sdp-thm-6} is the matrix-valued MacWilliams identity.
Finally, since \(\C\) detects every error in \(\E_\xi\) for \(\xi\in\mathcal D\),
the equality criterion of Theorem~\ref{thm:matrix-KL-equality} gives
\[
A_\xi(P)=K\,B_\xi(P)
\qquad\text{for all }\xi\in\mathcal D, \numberthis
\]
which establishes \eqref{eq:sdp-thm-7}. When \(V\) is irreducible the trivial
sector is one-dimensional, and \(A_{\mathbf1}=K^2/N\), \(B_{\mathbf1}=K/N\)
follow from Corollary~\ref{cor:code-identities} applied to \(X=P\).
\end{proof}

\begin{corollary}
\label{cor:intrinsic-sdp-infeasibility}
Fix a positive integer \(K\) and a prescribed subset \(\mathcal D\) of detected
isotypic sectors.
If the system \eqref{eq:sdp-thm-1}--\eqref{eq:sdp-thm-7} has no solution, then
there does not exist a code subspace \(\C\subseteq V\) of dimension \(K\)
capable of detecting every error in each sector \(\E_\xi\) for
\(\xi\in\mathcal D\).
Equivalently, infeasibility of this semidefinite system certifies that no
intrinsic code with the specified parameters exists.
\end{corollary}

The infeasibility criterion is turned into a dimension bound by the following
elementary monotonicity property, which replaces any direct ``maximization of
$K$'' and is used in both the semidefinite and linear settings.

\begin{lemma}[Detection passes to subcodes]
\label{lem:subcode-detection}
Let $\C'\subseteq\C\subseteq V$ with orthogonal projectors $P'$ and $P$.
If $\C$ detects an error $E$, then so does $\C'$.
Consequently, if the fixed-$(K_0{+}1)$ feasibility problem
\eqref{eq:sdp-thm-1}--\eqref{eq:sdp-thm-7} is infeasible for a given detected
set, then no code detecting that set has dimension exceeding $K_0$.
\end{lemma}

\begin{proof}
Since $\C'\subseteq\C$ we have $P'P=PP'=P'$. If $PEP=c_E P$, then
\[
    P'EP'=P'PEPP'=c_E\,P'P'=c_E\,P',
\]
so $\C'$ detects $E$. For the second statement, a code of dimension
$K\ge K_0{+}1$ detecting the given set would contain a $(K_0{+}1)$-dimensional
subcode which, by the first part, also detects that set, contradicting
infeasibility of the fixed-$(K_0{+}1)$ problem.
\end{proof}

Theorem~\ref{thm:intrinsic-sdp-feasibility} and
Corollary~\ref{cor:intrinsic-sdp-infeasibility} therefore reduce the existence
problem for intrinsic codes in the general equivariant setting to a finite
semidefinite feasibility problem.
For each fixed positive integer \(K\) these conditions form a semidefinite
feasibility problem; upper bounds on achievable intrinsic code dimension are
obtained by proving infeasibility at candidate integer values of \(K\) (see
Lemma~\ref{lem:subcode-detection}).

\section{An Equivariant Example with Multiplicity: Intrinsic Codes in $(2,2)$}
\label{sec:mult-su3-example}

We now illustrate the matrix-valued theory in the smallest example we have
found that combines three useful features: the conjugation representation has
genuine multiplicities (including multiplicities \(2\) and \(3\)); the full
intertwiner algebra and MacWilliams transform remain explicitly tractable; and
the representation contains a concrete \(5\)-dimensional subgroup-covariant
code. Thus the example is large enough for off-diagonal multiplicity data to
matter, but small enough for the resulting semidefinite bound to admit an exact
dual certificate. In contrast with the example of
Section~\ref{sec:su3-example}, where all irreducible sectors occur with
multiplicity one and the intrinsic feasibility problem is scalar, the present
case is genuinely matrix-valued.

\subsection{The representation $V \cong (2,2)$ and the decomposition of $\L(V)$}

Let $V \cong (2,2)$ be the irreducible $\SU 3$ representation of
dimension $27$. The decomposition of $\L(V)$ into irreducible
$\SU 3$-modules is
\[
\begin{aligned}
\L(V) \cong\;&
(0,0)\oplus 2(1,1)\oplus (3,0)\oplus (0,3)\oplus 3(2,2) \\
&\oplus 2(4,1)\oplus 2(1,4)\oplus (6,0)\oplus (0,6) \\
&\oplus 2(3,3)\oplus (5,2)\oplus (2,5)\oplus (4,4). 
\end{aligned}
\]

Several irreducible sectors occur with multiplicity greater than one.
For example, the adjoint representation $(1,1)$ occurs with multiplicity
two, and the representation $(2,2)$ with multiplicity three.
Accordingly,
\[
\L(V)\cong \bigoplus_{\xi} m_\xi\,\xi, \numberthis
\]
with $m_\xi>1$ for several $\xi$, and the intertwiner algebra
\[
\Hom_{\SU 3}(\L(V),\L(V)) \numberthis
\]
is noncommutative. Consequently, the intrinsic enumerators are
matrix-valued: for each irreducible sector $\xi$ with multiplicity
$m_\xi$, the quantities
\[
A_\xi,B_\xi \in M_{m_\xi}(\CC) \numberthis
\]
are positive semidefinite matrices.

\subsection{The isotypic adjoint-order error model}
\label{subsec:22-depth-model}

For this example the depth function is the adjoint-order filtration. For an
irreducible \(\SU 3\)-module \(\xi\), set
\[
\operatorname{ord}_{\mathrm{Ad}}(\xi)
:=
\min\{r\ge0:\xi\subseteq (1,1)^{\otimes r}\},
\qquad \operatorname{ord}_{\mathrm{Ad}}((0,0))=0.
\numberthis
\]
Restricted to the irreducible types occurring in \(\L(V)\), this gives
\[
\begin{array}{c|l}
\text{depth} & \text{irreducible types}\\ \hline
0 &(0,0),\\
1 &(1,1),\\
2 &(2,2),(3,0),(0,3),\\
3 &(4,1),(1,4),(3,3),\\
4 &(2,5),(5,2),(6,0),(0,6),(4,4).
\end{array}
\numberthis
\]
Our error sectors are the \emph{full isotypic components}. In particular, a
depth-$2$ code in this section detects the entire adjoint isotypic component
$(1,1)\oplus(1,1)$, not merely the distinguished adjoint copy arising from
the differential of the $\SU3$ action. This is the full-isotypic
adjoint-order error model. Intrinsically, this depth is defined independently
of any particular physical realization. As shown in
Section~\ref{subsec:22-significance}, in the canonical four-qutrit
realization the full adjoint isotypic component is precisely the compression
of the physical weight-one error space, so depth two coincides there with
ordinary distance two.

\subsection{A $5$-dimensional intrinsic code via restriction to $3.A_6$}

We briefly recall the code from \cite{usIntrinsicCodes}. Consider
$3.A_6 \subset \SU 3$, the triple cover of the alternating group
$A_6$. Under restriction to $3.A_6$, the irreducible representation
$V \cong (2,2)$ decomposes into irreducible $3.A_6$-modules as described in the appendix of \cite{us3}
\[
V\downarrow_{3.A_6}\cong \chi_6 \oplus \chi_7 \oplus \chi_{11} \oplus \chi_{12}. \numberthis
\]
In particular, here we will consider this first $5$-dimensional irreducible constituent
$\chi_6$. That is, we consider the intrinsic code $ \Cint:=\chi_6\subset V $. Then $ \dim \Cint = K=5 $. This description is purely in terms of
the representation-theoretic decomposition of $V$ under the subgroup
$3.A_6$, so that the code projector is exactly the projector onto the $ \chi_6 $ isotypic subspace.

\subsection{Semidefinite feasibility conditions}

The intrinsic enumerators
\[
A_\xi,B_\xi \in M_{m_\xi}(\CC) \numberthis
\]
satisfy the matrix-valued MacWilliams identity
\[
\mathrm{vec}(B)=M\,\mathrm{vec}(A), \numberthis
\]
with the coefficient matrix \(M\) of \eqref{eq:M-mult-def}, together with the
semidefinite constraints
\[
A_\xi \succeq 0,
\qquad
B_\xi\succeq 0,
\qquad
A_\xi \preceq K\,B_\xi, \numberthis
\]
the normalization conditions
\[
\sum_\xi \Tr(A_\xi)=K,
\qquad
\sum_\xi \Tr(B_\xi)=K^2, \numberthis
\]
and, since \(V\) is irreducible of dimension \(N=27\),
\[
A_{(0,0)}=\frac{K^2}{27},
\qquad
B_{(0,0)}=\frac{K}{27}. \numberthis
\]

These constraints are genuinely block-structured. For example,
\[
A_{(1,1)},\,B_{(1,1)}\in M_2(\CC),
\qquad
A_{(2,2)},\,B_{(2,2)}\in M_3(\CC). \numberthis
\]
The depth-$2$ detection condition imposes, on the adjoint sector,
\[
A_{(1,1)}=K\,B_{(1,1)} \numberthis
\]
(Theorem~\ref{thm:matrix-KL-equality}). This is the \emph{only} depth-$2$
detection constraint; we emphasize that it does \emph{not} in general force
\(A_{(1,1)}=0\), although the specific code \(\chi_6\) below happens to satisfy
\(A_{(1,1)}=B_{(1,1)}=0\).

For each fixed \(K\), the above is a semidefinite feasibility problem over the
block variables \(\{A_\xi,B_\xi\}\).

\paragraph{An explicit feasible code at $K=5$.}
The \(3.A_6\) constituent \(\Cint=\chi_6\) is an explicit depth-$2$ code of
dimension \(K=5\); its enumerator table (Section~\ref{subsec:22-table}) is a
feasible point of the \(K=5\) problem.

\paragraph{Non-uniqueness of the $K=5$ enumerator.}
The \(K=5\), depth-$2$ feasibility set is not a single point. In the coefficient
basis used for the computation, an explicit one-parameter family of feasible
projector enumerators is
\begin{equation}
\label{eq:22-K5-family}
\begin{aligned}
A_{(0,0)}&=\tfrac{25}{27}, &
A_{(2,5)}=A_{(5,2)}&=\tfrac53-\tfrac32 t, \\
A_{(6,0)}=A_{(0,6)}&=t+\tfrac{10}{27}, &
A_{(4,4)}&=t,
\end{aligned}
\qquad 0\le t\le\tfrac{10}{9},
\end{equation}
with every other \(A\)-block zero; at \(t=\tfrac{10}{9}\) this is the enumerator
of \(\chi_6\), which is thus an endpoint of the family. We therefore make no
claim that the feasible enumerator at \(K=5\) is unique: unlike the three
scalar examples above, whose feasible enumerators are forced, the matrix-valued
feasibility constraints here do not uniquely determine the \(K=5\) enumerator.

\begin{proposition}[Sharp bound for the isotypic adjoint-order model]
\label{prop:22-Kle5}
Every code in the \(\SU 3\) irrep \(V\cong(2,2)\) that detects the full adjoint
isotypic component---equivalently, every depth-\(2\) code for the error model of
Section~\ref{subsec:22-depth-model}---has dimension \(K\le5\). The bound is
attained by the \(3.A_6\) constituent \(\chi_6\).
\end{proposition}

\begin{proof}
By Theorem~\ref{thm:intrinsic-sdp-feasibility}, a depth-$2$ code of dimension
\(K\) yields Hermitian blocks \(A_\xi=A_\xi(P)\succeq0\) and
\(B_\xi=B_\xi(P)\succeq0\) with \(\mathrm{vec}(B)=M\,\mathrm{vec}(A)\), the
normalizations \(A_{(0,0)}=K^2/27\) and \(\sum_\xi\Tr A_\xi=K\), and the
detected-adjoint equality \(A_{(1,1)}=K\,B_{(1,1)}\)
(Theorem~\ref{thm:matrix-KL-equality}).
Appendix~\ref{app:22-certificate} exhibits positive definite matrices
\(Y_\xi\) for \(\xi\in S:=\{(1,1),(2,2),(0,3),(3,0),(4,1),(1,4),(3,3)\}\) and
\(Z=\operatorname{diag}(9/35,1)\succ0\), together with the exact identity, valid
for all \(A\) with \(\mathrm{vec}(B)=M\,\mathrm{vec}(A)\),
\begin{equation}
\label{eq:22-supporting-identity}
\sum_{\xi\in S}\Tr(Y_\xi A_\xi)
=
-\tfrac{16}{7}A_{(0,0)}
+\tfrac{80}{189}\sum_\xi\Tr A_\xi
-\Tr\!\bigl[Z\bigl(A_{(1,1)}-5B_{(1,1)}\bigr)\bigr].
\end{equation}
Substituting the depth-$2$ relations and
\(A_{(1,1)}-5B_{(1,1)}=(K-5)B_{(1,1)}\) gives
\begin{equation}
\label{eq:22-K-factor}
\sum_{\xi\in S}\Tr(Y_\xi A_\xi)
=
-(K-5)\Bigl(\tfrac{16K}{189}+\Tr\bigl(Z\,B_{(1,1)}\bigr)\Bigr).
\end{equation}
The left-hand side is nonnegative because each \(Y_\xi\succ0\) and
\(A_\xi\succeq0\); on the right, \(Z\succ0\) and \(B_{(1,1)}\succeq0\) make the
parenthesized factor strictly positive for \(K>0\). Hence \(K>5\) forces a
strictly negative right-hand side, a contradiction, so \(K\le5\). The
constituent \(\chi_6\) has dimension \(5\) and depth \(2\), so the bound is
attained.
\end{proof}

\begin{remark}
The certificate is stronger than fixed-\(K=6\) infeasibility: it uses only the
positivity \(A_\xi\succeq0\) and \(B_{(1,1)}\succeq0\), the MacWilliams identity,
and the universal normalizations, and does not invoke the full Loewner family
\(A_\xi\preceq K B_\xi\) or \(\sum_\xi\Tr B_\xi=K^2\). At \(K=5\),
\eqref{eq:22-K-factor} forces \(\sum_{\xi\in S}\Tr(Y_\xi A_\xi)=0\), hence
\(A_\xi=0\) for every \(\xi\in S\); this is consistent with the family
\eqref{eq:22-K5-family}, whose support avoids \(S\). Thus \(K=5\) is an exposed
boundary face of the matrix-valued feasibility region, and sharpness does not
require enumerator uniqueness.
\end{remark}

\subsection{An explicit solution of the semidefinite constraints}
\label{subsec:22-table}

Before displaying the enumerators we record two facts used in the computation.
First, the coefficient matrix \(M\) for this example is real, which justifies
restricting the search to real-symmetric blocks.

\begin{lemma}[Real-symmetric reduction]
\label{lem:real-sym-reduction}
Suppose the coefficient matrix \(M\) is real. If the complex-Hermitian
feasibility problem \eqref{eq:sdp-thm-1}--\eqref{eq:sdp-thm-7} has a solution,
then it has one in which every \(A_\xi\) and \(B_\xi\) is real symmetric.
\end{lemma}

\begin{proof}
If \((A,B)\) is feasible then, because all coefficients are real, so is
\((\overline A,\overline B)\); by convexity of the constraints
\(\bigl(\tfrac{A+\overline A}{2},\tfrac{B+\overline B}{2}\bigr)\) is feasible,
and it is real symmetric since \(A,B\) are Hermitian. All Loewner inequalities
and the affine constraints are preserved under this averaging.
\end{proof}

Second, we record a structural caveat on the error model, valid throughout the
matrix-valued theory.

\begin{remark}[Isotypic versus irreducible sectors]
\label{rmk:isotypic-caveat}
The isotypic components $\E_\xi$ are canonical, whereas a splitting into
equivalent irreducible copies $\E_{\xi\alpha}$ is not. We use whole
isotypic components because they are the realization-independent sectors
controlled by the Schur bootstrap \cite{usIntrinsicCodes}. Additional operator structure may
nevertheless distinguish particular irreducible copies; for example, the
differential of the $\SU3$ action selects a distinguished adjoint submodule
inside $\E_{(1,1)}\cong2(1,1)$. Such finer error models can also be
meaningful, but they are not canonical from the isotypic decomposition
alone.
\end{remark}

The semidefinite feasibility conditions admit an explicit solution
corresponding to the code $\Cint=\chi_6$.
Within each multiplicity space we choose the orthonormal error basis
that diagonalises $B_\xi$; since $A_\xi=0$ for every sector with
$m_\xi>1$, this simultaneously diagonalises both enumerators.
The resulting intrinsic enumerators are:

\[
\resizebox{\textwidth}{!}{$
\begin{array}{c|c|c|ccc|ccc|ccccc}
\text{Rep}
  & \dynkin{0,0} & \dynkin{1,1}
  & \dynkin{2,2} & \dynkin{0,3} & \dynkin{3,0}
  & \dynkin{4,1} & \dynkin{1,4} & \dynkin{3,3}
  & \dynkin{2,5} & \dynkin{5,2} & \dynkin{6,0} & \dynkin{0,6} & \dynkin{4,4}
\\
\hline
\text{Depth} & 0 & 1 & 2 & 2 & 2 & 3 & 3 & 3 & 4 & 4 & 4 & 4 & 4
\\ \hline
A
  & \tfrac{25}{27}
  & \smqty(0&0\\0&0)
  & \smqty(0&0&0\\0&0&0\\0&0&0)
  & 0 & 0
  & \smqty(0&0\\0&0)
  & \smqty(0&0\\0&0)
  & \smqty(0&0\\0&0)
  & 0 & 0
  & \tfrac{40}{27} & \tfrac{40}{27} & \tfrac{10}{9}
\\[6pt]
B
  & \tfrac{5}{27}
  & \smqty(0&0\\0&0)
  & \smqty(\alpha_+&0&0\\0&0&0\\0&0&\alpha_-)
  & \tfrac{125}{189} & \tfrac{125}{189}
  & \smqty(\tfrac{25}{18}&0\\[3pt]0&\tfrac{25}{36})
  & \smqty(\tfrac{25}{18}&0\\[3pt]0&\tfrac{25}{36})
  & \smqty(\tfrac{56}{27}&0\\[3pt]0&\tfrac{520}{189})
  & \tfrac{73}{28} & \tfrac{73}{28}
  & \tfrac{29}{27} & \tfrac{29}{27} & \tfrac{235}{54}
\end{array} 
$} \numberthis
\]

\noindent
Here
\[
  \alpha_\pm = \frac{117 \pm \sqrt{3189}}{84} \numberthis
\]
are the eigenvalues of the $\{\mu{=}1,\mu{=}3\}$ sub-block of $B_{(2,2)}$
(numerically $\alpha_+\approx 2.065$ and $\alpha_-\approx 0.721$),
and the $\mu{=}2$ entry is zero, identifying an error direction
invisible to the code.
All other off-diagonal entries are zero, and $(4,1)$ and $(1,4)$
carry identical diagonal entries reflecting the conjugate symmetry of the code.



\subsection{Significance of the $(2,2)$ example}
\label{subsec:22-significance}

The example above serves first as a proof of concept for the genuinely
matrix-valued theory.  Proposition~\ref{prop:22-Kle5} gives the upper bound
\[
    K\leq 5
\]
from the noncommutative MacWilliams constraints, while the
$3.A_6$ constituent $\Cint=\chi_6$ has $K=5$ and satisfies the required
depth-$2$ conditions.  Thus the semidefinite relaxation can be sharp: in
this example the matrix-valued enumerator constraints recover the exact
optimal code dimension.

The example also has a direct multiqudit interpretation through the Schur
bootstrap of~\cite{usIntrinsicCodes}.  The canonical embedding of the
$\SU3$ irrep $V=(2,2)$ is
\[
    \iota:V\hookrightarrow
    \CC^3\otimes\CC^3\otimes\overline{\CC^3}
    \otimes\overline{\CC^3},
    \numberthis
\]
where $V$ occurs with multiplicity one.  Consequently, every intrinsic
depth-$2$ code in $V$ gives a four-qutrit code of ordinary multi-qudit distance at
least two.  In this particular realization the converse holds as well.
If $\Phi(E)=\iota^\dagger E\iota$ denotes compression to $V$, then the
compressed traceless weight-one error space is
\[
    \operatorname{span}\left\{
        \Phi(X^{(r)}):
        X\in\mathfrak{sl}_3(\CC),\ 1\leq r\leq4
    \right\}
    =
    \E_{(1,1)}
    \cong 2(1,1),
    \numberthis
\]
with the fundamental or anti-fundamental action understood on the
corresponding tensor factor.  Thus the independent one-site errors fill
the entire adjoint isotypic component.  Detection of the full
$\E_{(1,1)}$ intrinsic sector is therefore equivalent here to detection of arbitrary
weight-one errors in the multi-qutrit realization.  It follows that the bound $K\leq5$ is also the sharp
bound for distance-$2$ codes supported in this canonical $(2,2)$ sector of
four qutrits, and $\chi_6$ attains it.

There is a useful comparison with the multiplicity-free regime at the same
block length and local dimension.  The symmetric four-qutrit sector
\[
    V_{\mathrm{sym}}
    =
    \Sym^4(\CC^3)
    \cong(4,0)
\]
has multiplicity-free conjugation representation by
Proposition~\ref{prop:sym-power-mult-free}, so its enumerator bound is a
scalar linear program.  In this case that program gives the sharp
distance-$2$ bound
\[
    K\leq3.
    \numberthis
\]
For completeness, this follows already from the $K=4$ feasibility
conditions.  Writing the normalized projector enumerator as
$\widetilde A=(1,\widetilde A_1,\ldots,\widetilde A_4)$, the normalization
and adjoint-detection equations reduce to
\[
    \widetilde A_1+\widetilde A_2+\widetilde A_3+\widetilde A_4
    =\frac{11}{4},
\]
and, after eliminating $\widetilde A_4$,
\[
    49\widetilde A_1
    +64\widetilde A_2
    +36\widetilde A_3
    +16
    =0.
\]
This is impossible by nonnegativity, so the $K=4$ linear program is
infeasible; Lemma~\ref{lem:subcode-detection} then gives $K\leq3$.
The bound is attained by a permutation-invariant
$((4,3,2))_3$ code \cite{gross3}.  Thus, already for four qutrits and distance two, the
non-multiplicity-free $(2,2)$ sector supports a $5$-dimensional code $((4,5,2))_3$,
whereas the natural multiplicity-free symmetric sector supports at most
dimension $3$.  Multiplicity is therefore not merely a technical
complication of the formalism: the genuinely matrix-valued regime can
contain larger optimal codes than its multiplicity-free counterpart.

More broadly, this example illustrates why it is useful to formulate and
optimize quantum codes intrinsically.  The Schur bootstrap
of~\cite{usIntrinsicCodes} implies that Knill--Laflamme conditions imposed
on full isotypic sectors of $V$ propagate to every equivariant realization
of $V$.  One may therefore solve the enumerator problem once in the
typically much smaller representation space $V$, and transfer the resulting
code and bounds simultaneously to its physical realizations.  For standard
fundamental multiqudit realizations, adjoint order agrees with physical
error weight and intrinsic depth is ordinary code distance; other
equivariant realizations include heterogeneous multiqudit, bosonic, and
molecular settings in which the same intrinsic protection conditions
continue to apply.  This realization-independence is also a principal
reason for using full isotypic error sectors: they are canonical from
$(G,V)$ and are precisely the sectors controlled uniformly by equivariant
compression.

The $(2,2)$ example is intended as the first test of this noncommutative
regime rather than an exhaustive study of it.  A natural direction for
future work is therefore to apply the matrix-valued MacWilliams transform
and its SDP to broader families of non-multiplicity-free irreducible
representations of $\SU q$, to search for exact dual certificates and
explicit codes attaining them, and to determine how often multiplicity
permits code parameters that cannot occur in comparable
multiplicity-free sectors.

\appendices

\makeatletter
\renewcommand{\theequation}{A\arabic{equation}}
\renewcommand{\thetable}{A\arabic{table}}
\renewcommand{\thefigure}{A\arabic{figure}}
\renewcommand{\thelemma}{A\arabic{lemma}}
\renewcommand{\theremark}{A\arabic{remark}}
\renewcommand{\theproposition}{A\arabic{proposition}}
\renewcommand{\thecorollary}{A\arabic{corollary}}
\renewcommand{\thetheorem}{A\arabic{theorem}}
\setcounter{table}{0}
\setcounter{figure}{0}
\setcounter{lemma}{0}
\setcounter{remark}{0}
\setcounter{corollary}{0}
\setcounter{proposition}{0}
\setcounter{theorem}{0}
\setcounter{equation}{0}

\section{The Multiplicity--Free Specialization}
\label{sec:mf-macwilliams}

Section~\ref{sec:mult} developed equivariant MacWilliams theory for an
arbitrary finite-dimensional unitary representation \(V\) of a group \(G\),
allowing arbitrary multiplicities in the conjugation representation on
\(\L(V)\cong V^*\otimes V\).
This appendix records the special case in which that conjugation
representation is \emph{multiplicity--free}, and then carries out the two
concrete computations that this case makes possible: the \(\SU2\) transform in
terms of Wigner \(6j\)-symbols, and the wreath-product transform that
reproduces the quantum Krawtchouk matrix.
The symmetric-power and permutation-invariant examples that use this
specialization are collected in Appendix~\ref{sec:mf-examples}.

Nothing in Sections~\ref{subsec:mf-collapse}--\ref{subsec:wreath-krawtchouk} below
is logically new: every object and identity is the \(m_\xi\equiv1\) instance of
a corresponding construction in Section~\ref{sec:mult}.
As discussed in Section~\ref{subsec:relation-to-prior}, this
multiplicity--free theory has also been developed previously in the language of
quantum metric spaces --- the \(\SU2\) case by
Bumgardner~\cite{bumgardner2012codeswastmetricspacestheory}, and the general
linear-programming framework by
Okada~\cite{okada2025quantumanalogdelsarteslinear}.
We include the specialization only to fix the scalar notation used in the
examples, to exhibit the commutative limit of the matrix-valued theory
explicitly, and to record the two closed-form transforms, which do not appear
in the general development.

\subsection{The multiplicity--free collapse}
\label{subsec:mf-collapse}

Assume the conjugation representation on \(\L(V)\) is multiplicity--free, so
that its isotypic decomposition
\begin{equation}\label{eq:mf-decomp}
    \L(V)=\bigoplus_{\xi}\E_\xi
\end{equation}
has every nonzero summand \(\E_\xi\) irreducible and occurring with
multiplicity \(m_\xi=1\).
Every construction of Section~\ref{sec:mult} then collapses along its
multiplicity index, and we obtain the corresponding scalar theory with no
further work.
We record the collapse once, in the order in which the objects were introduced
in the main text.

\paragraph{Enumerator data.}
With each multiplicity space one-dimensional, the triple index
\((\xi,\alpha,\beta)\) reduces to the single label \(\xi\).
The projector matrix unit and generalized twirl of Section~\ref{sec:mult}
become the ordinary projector and twirling superoperators
\begin{equation}\label{eq:mf-collapse-ops}
    \Pi_{\xi11}=\proj_\xi,
    \qquad
    \Twirl_{\xi11}=\twirl_\xi,
\end{equation}
where, for an orthonormal basis \(\B_\xi\) of \(\E_\xi\),
\[
    \proj_\xi(X)=\sum_{E\in\B_\xi}\ip{E,X}\,E,
    \qquad
    \twirl_\xi(X)=\sum_{E\in\B_\xi}E^\dagger X E .
\]
The matrix-valued enumerators of Section~\ref{sec:mult} become the scalars
\begin{align}
    A_\xi(X_1,X_2)
    &:= \ip{X_1,\proj_\xi(X_2)}
     = \sum_{E\in\B_\xi}\Tr(X_1^\dagger E)\Tr(X_2E^\dagger),
    \label{eq:mf-a-def}\\
    B_\xi(X_1,X_2)
    &:= \ip{X_1,\twirl_\xi(X_2)}
     = \sum_{E\in\B_\xi}\Tr(X_1^\dagger E^\dagger X_2 E),
    \label{eq:mf-b-def}
\end{align}
and we again write \(A_\xi(X,X)\), \(B_\xi(X,X)\) for the associated quadratic
enumerators.
These are exactly the sesquilinear forms of Section~\ref{sec:intrinsic-forms}
attached to the decomposition~\eqref{eq:mf-decomp}, so all of their properties
are already available.

\paragraph{Positivity, normalization, and detection.}
Specializing Lemma~\ref{lem:block-psd}, Lemma~\ref{lem:matrix-KL}, and
Theorem~\ref{thm:matrix-KL-equality} to \(m_\xi=1\) turns every Loewner
inequality into a scalar inequality and every trace into the scalar itself.
Equivalently, one may read these facts directly off
Lemma~\ref{lem:basic-identities}, Lemma~\ref{lem:ai-bi-ineq}, and
Corollary~\ref{cor:code-identities} of Section~\ref{sec:intrinsic-forms}.
We collect the result for later reference.

\begin{corollary}
\label{cor:mf-basic-identities}
Let the conjugation representation on \(\L(V)\) be multiplicity--free, and let
\(\C\subseteq V\) be a code of dimension \(K:=\dim\C\) with projector \(P\).
Then for every \(\xi\),
\[
    A_\xi(P,P)\ge0,
    \qquad
    \sum_\xi A_\xi(P,P)=K,
    \qquad
    \sum_\xi B_\xi(P,P)=K^2,
\]
and
\[
    A_\xi(P,P)\le K\,B_\xi(P,P),
\]
with equality if and only if \(PEP=c_E P\) for all \(E\in\B_\xi\), i.e.\ if and
only if \(\C\) detects the sector \(\E_\xi\).
If \(\E_0=\mathrm{span}\{I\}\) then \(A_0(P,P)=K^2/N\) and \(B_0(P,P)=K/N\).
\end{corollary}

\begin{proof}
Specialize Corollary~\ref{cor:code-identities} (equivalently, set
\(m_\xi=1\) in Lemma~\ref{lem:block-psd} and
Theorem~\ref{thm:matrix-KL-equality}).
\end{proof}

\paragraph{The intertwiner algebra and the transform.}
By Schur's lemma each one-dimensional multiplicity space contributes a
one-dimensional endomorphism algebra, so
\[
    \Hom_G(\L(V),\L(V))\cong\bigoplus_\xi\CC
\]
is commutative.
The two orthonormal bases of Section~\ref{sec:mult} collapse to the scalar
families
\[
    \Bigl\{d_\xi^{-1/2}\proj_\xi\Bigr\}_\xi,
    \qquad
    \Bigl\{d_\xi^{-1/2}\twirl_\xi\Bigr\}_\xi ,
    \qquad d_\xi:=\dim\xi,
\]
each an orthonormal basis of this commutative algebra.
The unitary change of basis between them is the scalar matrix \(U\) obtained by
setting \(m_\xi=1\) in Proposition~\ref{prop:block-macwilliams}, and the
associated coefficient matrix \(M\) of~\eqref{eq:M-mult-def} becomes the scalar
transform
\begin{equation}\label{eq:mf-M-def}
    M_{\xi\rho}=\sqrt{\tfrac{d_\xi}{d_\rho}}\,U_{\xi\rho},
    \qquad
    B_\xi(P,P)=\sum_\rho M_{\xi\rho}\,A_\rho(P,P),
\end{equation}
which is the scalar case of the matrix MacWilliams identity
(Proposition~\ref{prop:matrix-macwilliams}).
The weighted-orthogonality relation is inherited directly.

\begin{proposition}[Weighted orthogonality]
\label{prop:mf-weighted-orthogonality}
With \(D:=\mathrm{diag}(d_\xi)\), the scalar MacWilliams matrix \(M\)
of~\eqref{eq:mf-M-def} satisfies
\[
    M\,D\,M^{\mathsf T}=D,
\]
equivalently \(D^{-1/2}MD^{1/2}=U\) is unitary.
\end{proposition}

\begin{proof}
This is the multiplicity--free case of the weighted-orthogonality relation
\eqref{eq:M-mult-weighted-orth}; alternatively substitute
\(M=D^{1/2}UD^{-1/2}\) and use \(UU^\dagger=I\).
\end{proof}

The matrix \(M\) plays the role of the Krawtchouk transform of classical coding
theory: it carries the projector enumerator to the dual (twirl) enumerator, and
Proposition~\ref{prop:mf-weighted-orthogonality} shows it is orthogonal with
respect to the sector dimensions \(d_\xi\).
The trivial sector gives a constant first row, since specializing
Lemma~\ref{lem:trivial-row-audit} to an irreducible \(V\) of dimension \(N\)
gives \(M_{0\rho}=1/N\) for every \(\rho\).

\paragraph{The linear program.}
Finally, the semidefinite feasibility theorem collapses to a linear one.

\begin{theorem}[Intrinsic LP feasibility]
\label{thm:intrinsic-lp-feasibility}
Let \(V\) be a finite-dimensional unitary representation of \(G\) whose
conjugation representation on \(\L(V)=\bigoplus_\xi\E_\xi\) is
multiplicity--free.
Suppose there is a code \(\C\subseteq V\) of dimension \(K:=\dim\C\) detecting
every sector \(\E_\xi\) for \(\xi\in\mathcal D\).
Then there exist real numbers \(\{A_\xi\},\{B_\xi\}\) with
\begin{align}
    A_\xi &\ge 0
        &&\text{for all }\xi,\label{eq:lp-thm-1}\\
    A_\xi &\le K\,B_\xi
        &&\text{for all }\xi,\label{eq:lp-thm-2}\\
    \textstyle\sum_\xi A_\xi &= K,\label{eq:lp-thm-3}\\
    \textstyle\sum_\xi B_\xi &= K^2,\label{eq:lp-thm-4}\\
    B_\xi &= \textstyle\sum_\rho M_{\xi\rho}A_\rho
        &&\text{for all }\xi,\label{eq:lp-thm-5}\\
    A_\xi &= K\,B_\xi
        &&\text{for all }\xi\in\mathcal D.\label{eq:lp-thm-6}
\end{align}
One may take \(A_\xi=A_\xi(P,P)\) and \(B_\xi=B_\xi(P,P)\).
\end{theorem}

\begin{proof}
Set \(m_\xi=1\) in Theorem~\ref{thm:intrinsic-sdp-feasibility}: the Hermitian
blocks \(A_\xi,B_\xi\) become nonnegative reals,
\eqref{eq:sdp-thm-1}--\eqref{eq:sdp-thm-2} become
\eqref{eq:lp-thm-1}, the Loewner inequality \eqref{eq:sdp-thm-3} becomes
\eqref{eq:lp-thm-2}, the trace normalizations \eqref{eq:sdp-thm-4}--\eqref{eq:sdp-thm-5}
become \eqref{eq:lp-thm-3}--\eqref{eq:lp-thm-4}, the vectorized identity
\eqref{eq:sdp-thm-6} becomes the scalar transform \eqref{eq:lp-thm-5}, and the
detection equalities \eqref{eq:sdp-thm-7} become \eqref{eq:lp-thm-6}.
The values \(A_\xi(P,P),B_\xi(P,P)\) are feasible by
Corollary~\ref{cor:mf-basic-identities}.
\end{proof}

\begin{corollary}
\label{cor:intrinsic-lp-infeasibility}
Fix \(K\) and a detected set \(\mathcal D\).
If the system \eqref{eq:lp-thm-1}--\eqref{eq:lp-thm-6} is infeasible, then no
code of dimension \(K\) detects every sector \(\E_\xi\) for \(\xi\in\mathcal D\).
Combined with Lemma~\ref{lem:subcode-detection}, infeasibility of the
fixed-\((K_0{+}1)\) system implies that no such code has dimension exceeding
\(K_0\).
\end{corollary}

\begin{proof}
The multiplicity--free case of
Corollary~\ref{cor:intrinsic-sdp-infeasibility}, together with
Lemma~\ref{lem:subcode-detection}.
\end{proof}

\paragraph{Normalized coordinates.}
For the examples we use the coding-theoretic normalization of
Section~\ref{sec:sl-recovery},
\[
    \widetilde A_\xi:=\frac{N}{K^2}A_\xi(P,P),
    \qquad
    \widetilde B_\xi:=\frac{N}{K}B_\xi(P,P),
\]
under which \(\widetilde A_0=\widetilde B_0=1\) (using
\(\E_0=\mathrm{span}\{I\}\), which exists because \(V\) is irreducible), the
MacWilliams identity reads \(\widetilde B=KM\widetilde A\), and each detected
sector contributes the coding-theoretic equality \(\widetilde A_\xi=\widetilde
B_\xi\).
These are the same normalized enumerators that, in the wreath-product case of
Section~\ref{sec:sl-recovery}, coincide with the Shor--Laflamme enumerators
\(A_w^{\mathrm{SL}},B_w^{\mathrm{SL}}\).

\subsection{The \texorpdfstring{$\SU2$}{SU(2)} transform via \texorpdfstring{$6j$}{6j}-symbols}
\label{subsec:su2-macwilliams}

We now specialize the transform~\eqref{eq:mf-M-def} to \(G=\SU2\), where it has
a closed form; this case was first treated by
Bumgardner~\cite{bumgardner2012codeswastmetricspacestheory}.
Let \(V\cong V_j\) be the spin-\(j\) irreducible representation.
Under conjugation,
\begin{equation}\label{eq:L-su2-decomp}
    \L(V)\cong V_j\otimes V_j^*\cong\bigoplus_{k=0}^{2j}V_k,
\end{equation}
where \(V_k\) is the \((2k+1)\)-dimensional space of rank-\(k\) irreducible
tensor operators; the decomposition is multiplicity--free.
For each \(k\) let \(\Pi_k\) be the projector onto \(V_k\) and
\begin{equation}\label{eq:Twirlk-def}
    \Twirl_k(X):=\sum_{q=-k}^{k}\bigl(T^{(k)}_q\bigr)^\dagger X\,T^{(k)}_q,
\end{equation}
for any orthonormal basis \(\{T^{(k)}_q\}_{q=-k}^{k}\) of \(V_k\) (for example
the rank-\(k\) spherical tensors~\cite{quantumtheoryangularmomentum}); the
definition is basis-independent.
By Section~\ref{subsec:mf-collapse} the normalized families
\[
    \Bigl\{(2k+1)^{-1/2}\Pi_k\Bigr\}_{k=0}^{2j},
    \qquad
    \Bigl\{(2k+1)^{-1/2}\Twirl_k\Bigr\}_{k=0}^{2j}
\]
are dual orthonormal bases of \(\Hom_{\SU2}(\L(V),\L(V))\), and the transform is
\begin{equation}\label{eq:U-def-su2}
    [U]_{k_1k_2}
    =\big\langle(2k_1+1)^{-1/2}\Pi_{k_1},\,(2k_2+1)^{-1/2}\Twirl_{k_2}\big\rangle,
    \qquad 0\le k_1,k_2\le 2j.
\end{equation}

\begin{remark}[Index convention]
\label{rmk:su2-index}
The general transform~\eqref{eq:M-mult-def} places the twirl (output) label
first, whereas~\eqref{eq:U-def-su2} places the projector label first.
This is immaterial here: by the row-exchange symmetry of the \(6j\)-symbol
in~\eqref{eq:U-6j}, \([U]_{k_1k_2}\) is symmetric, and the unnormalized matrix
\(M\) of Corollary~\ref{cor:M-6j} is written in the output-as-row convention
\(B_{k_1}=\sum_{k_2}M_{k_1k_2}A_{k_2}\), consistent with the general theory.
\end{remark}

\begin{proposition}[Scalar action of the twirling maps]
\label{prop:twirl-scalar}
For each \(0\le k_1,k_2\le 2j\), the map \(\Twirl_{k_2}\) acts on the summand
\(V_{k_1}\) in~\eqref{eq:L-su2-decomp} as a scalar,
\[
    \Twirl_{k_2}\big|_{V_{k_1}}=c_{k_1k_2}\,\mathrm{id}_{V_{k_1}},
\]
and
\begin{equation}\label{eq:U-vs-c}
    [U]_{k_1k_2}=\sqrt{\tfrac{2k_1+1}{2k_2+1}}\;c_{k_1k_2}.
\end{equation}
\end{proposition}

\begin{proof}
Equivariance of \(\Twirl_{k_2}\) and Schur's lemma give the scalar action.
Then \eqref{eq:U-def-su2} gives
\[
    [U]_{k_1k_2}
    =\frac{\Tr(\Pi_{k_1}^\dagger\Twirl_{k_2})}{\sqrt{(2k_1+1)(2k_2+1)}}
    =\frac{c_{k_1k_2}(2k_1+1)}{\sqrt{(2k_1+1)(2k_2+1)}},
\]
which is~\eqref{eq:U-vs-c}.
\end{proof}

\begin{proposition}[Closed form via \texorpdfstring{$6j$}{6j}-symbols]
\label{prop:su2-6j}
The scalars are
\begin{equation}\label{eq:c-6j}
    c_{k_1k_2}
    =(2k_2+1)(-1)^{2j+k_1+k_2}
    \begin{Bmatrix} j&j&k_1\\ j&j&k_2\end{Bmatrix},
\end{equation}
so the transform is
\begin{equation}\label{eq:U-6j}
    [U]_{k_1k_2}
    =(-1)^{2j+k_1+k_2}\sqrt{(2k_1+1)(2k_2+1)}
    \begin{Bmatrix} j&j&k_1\\ j&j&k_2\end{Bmatrix}.
\end{equation}
\end{proposition}

\begin{proof}
Fix a unit vector \(T^{(k_1)}_{q_1}\in V_{k_1}\).
By Proposition~\ref{prop:twirl-scalar},
\[
    c_{k_1k_2}
    =\big\langle T^{(k_1)}_{q_1},\Twirl_{k_2}(T^{(k_1)}_{q_1})\big\rangle
    =\sum_{q_2=-k_2}^{k_2}
     \Tr\!\Big((T^{(k_1)}_{q_1})^\dagger(T^{(k_2)}_{q_2})^\dagger
     T^{(k_1)}_{q_1}T^{(k_2)}_{q_2}\Big).
\]
By the Racah recoupling identity for irreducible tensor
operators~\cite{quantumtheoryangularmomentum},
\[
    \sum_{q_2=-k_2}^{k_2}(T^{(k_2)}_{q_2})^\dagger
    T^{(k_1)}_{q_1}T^{(k_2)}_{q_2}
    =(2k_2+1)(-1)^{2j+k_1+k_2}
     \begin{Bmatrix} j&j&k_1\\ j&j&k_2\end{Bmatrix}
     T^{(k_1)}_{q_1},
\]
and taking the inner product with \(T^{(k_1)}_{q_1}\) gives~\eqref{eq:c-6j}.
Substituting into~\eqref{eq:U-vs-c} gives~\eqref{eq:U-6j}.
\end{proof}

\begin{corollary}[Unnormalized \texorpdfstring{$\SU2$}{SU(2)} MacWilliams matrix]
\label{cor:M-6j}
In the unnormalized basis \(B_{k_1}=\sum_{k_2=0}^{2j}M_{k_1k_2}A_{k_2}\),
\begin{equation}\label{eq:M-6j}
    M_{k_1k_2}
    =(-1)^{2j+k_1+k_2}(2k_1+1)
     \begin{Bmatrix} j&j&k_1\\ j&j&k_2\end{Bmatrix}.
\end{equation}
\end{corollary}

\begin{proof}
Substitute~\eqref{eq:U-6j} into \(M_{k_1k_2}=\sqrt{(2k_1+1)/(2k_2+1)}\,[U]_{k_1k_2}\).
\end{proof}

The \(6j\)-symbol has a natural representation-theoretic reading.
Under \(V_j^*\cong V_j\) one has \(\L(V_j)\cong V_j\otimes V_j\), hence
\(\Hom_{\SU2}(\L(V_j),\L(V_j))\cong(V_j^{\otimes4})^{\SU2}\); the projector and
twirl bases correspond to two coupling schemes for four spins \(j\), and the
transform is the Racah recoupling between them.
The same matrix appears in~\cite{fan2026randomizedbenchmarkingsyntheticquantum}
with a randomized-benchmarking interpretation.
As a check, \eqref{eq:M-6j} with \(k_1=0\) gives the constant trivial row
\(M_{0k}=\tfrac{1}{2j+1}\), consistent with the trivial-row audit
(Lemma~\ref{lem:trivial-row-audit}) since \(N=\dim V_j=2j+1\).

\subsection{The wreath-product transform and the quantum Krawtchouk matrix}
\label{subsec:wreath-krawtchouk}

The second closed-form case is the non-entangling wreath product of
Section~\ref{sec:sl-recovery}, whose transform is the quantum Krawtchouk
matrix.
Let
\[
    V=(\CC^q)^{\otimes n},
    \qquad
    G=U(q)^n\rtimes S_n=U(q)\wr S_n .
\]
The central \(U(1)\)'s act trivially by conjugation, so one may replace
\(U(q)\) by \(\SU q\) without changing the error sectors.
On one site
\(\L(\CC^q)=\CC I\oplus\L_0(\CC^q)\cong\mathbf1\oplus\Ad\), and as recalled in
Section~\ref{sec:sl-recovery} the wreath product merges supports of equal size,
so the isotypic sectors are the Hamming-weight spaces
\[
    \E_w=\bigoplus_{\substack{S\subseteq[n]\\|S|=w}}\E_S ,
    \qquad 0\le w\le n,
    \qquad
    \L(V)=\bigoplus_{w=0}^n\E_w .
\]
This decomposition is multiplicity--free as a representation of
\(U(q)\wr S_n\): in the standard parametrization of irreducibles of the wreath
product, \(\E_w\) assigns the one-row partition \((n-w)\) to the trivial
single-site representation and \((w)\) to the adjoint, and these are pairwise
inequivalent.
Thus the specialization of this appendix applies.

For a single site write the two sectors as \(\E_0=\CC I\) and
\(\E_1=\L_0(\CC^q)\), and let \(M_{\mathrm{site}}\) be the one-site unnormalized
transform \(B_i=\sum_{j}[M_{\mathrm{site}}]_{ij}A_j\).

\begin{lemma}[One-site transform]
\label{lem:one-site-krawtchouk}
\begin{equation}\label{eq:one-site-quantum-macwilliams}
    M_{\mathrm{site}}
    =\frac{1}{q}
     \begin{pmatrix} 1 & 1\\ q^2-1 & -1\end{pmatrix}.
\end{equation}
\end{lemma}

\begin{proof}
The trivial sector \(\E_0=\CC I\) gives the constant first row
\(M_{0j}=1/q\) by the trivial-row audit
(Lemma~\ref{lem:trivial-row-audit}) with \(N=q\).
For the second row, \(\dim\E_0=1\) and \(\dim\E_1=q^2-1\), and the
weighted-orthogonality relation \(M_{\mathrm{site}}\,D\,M_{\mathrm{site}}^{\mathsf T}=D\)
of Proposition~\ref{prop:mf-weighted-orthogonality}, with
\(D=\mathrm{diag}(1,q^2-1)\), forces the second row to be \((q^2-1,-1)/q\).
\end{proof}

For \(n\) sites, before quotienting by \(S_n\), the support-refined transform is
the tensor power \(M_{\mathrm{site}}^{\otimes n}\); passing to the
\(S_n\)-invariant weight basis gives the Krawtchouk matrix.

\begin{proposition}[Wreath-product MacWilliams transform]
\label{prop:wreath-krawtchouk}
For \(V=(\CC^q)^{\otimes n}\) and \(G=U(q)\wr S_n\), the unnormalized intrinsic
MacWilliams matrix for the weight decomposition \(\L(V)=\bigoplus_{w=0}^n\E_w\)
has entries
\[
    M_{wi}=\frac{1}{q^n}\,\mathsf K^{(n,q)}_w(i),
    \qquad 0\le w,i\le n,
\]
where
\[
    \mathsf K^{(n,q)}_w(i)
    :=\sum_{r=0}^{w}(-1)^r(q^2-1)^{w-r}\binom{i}{r}\binom{n-i}{w-r}
\]
is the \(q\)-ary quantum Krawtchouk polynomial.
\end{proposition}

\begin{proof}
The support-refined transform is the \(n\)-fold tensor power of
\eqref{eq:one-site-quantum-macwilliams}.
Fix an input error of weight \(i\).
To produce an output of weight \(w\), choose \(r\) of the \(i\) nonidentity
input positions to stay in the traceless sector and \(w-r\) of the \(n-i\)
identity positions to enter it; \eqref{eq:one-site-quantum-macwilliams}
contributes \(-1\) per nonidentity-to-nonidentity position and \(q^2-1\) per
identity-to-nonidentity position, with global factor \(q^{-n}\).
Summing over \(r\) gives the stated entry.
\end{proof}

Combining Proposition~\ref{prop:wreath-krawtchouk} with the normalized relation
\(\widetilde B=KM\widetilde A\) and the identification
\(\widetilde A_i=A_i^{\mathrm{SL}}\), \(\widetilde B_w=B_w^{\mathrm{SL}}\) of
Section~\ref{sec:sl-recovery} recovers the ordinary Shor--Laflamme quantum
MacWilliams identity
\[
    B_w^{\mathrm{SL}}
    =\frac{K}{q^n}\sum_{i=0}^{n}\mathsf K^{(n,q)}_w(i)\,A_i^{\mathrm{SL}} .
\]
Thus the standard quantum weight enumerators and their Krawtchouk transform are
the intrinsic enumerator theory of the non-entangling wreath-product symmetry,
and the \(\SU2\) transform of Section~\ref{subsec:su2-macwilliams} is the
corresponding statement for a single spin-\(j\) irrep.

\makeatletter
\renewcommand{\theequation}{B\arabic{equation}}
\renewcommand{\thetable}{B\arabic{table}}
\renewcommand{\thefigure}{B\arabic{figure}}
\renewcommand{\thelemma}{B\arabic{lemma}}
\renewcommand{\theremark}{B\arabic{remark}}
\renewcommand{\theproposition}{B\arabic{proposition}}
\renewcommand{\thecorollary}{B\arabic{corollary}}
\renewcommand{\thetheorem}{B\arabic{theorem}}
\setcounter{table}{0}
\setcounter{figure}{0}
\setcounter{lemma}{0}
\setcounter{remark}{0}
\setcounter{corollary}{0}
\setcounter{proposition}{0}
\setcounter{theorem}{0}
\setcounter{equation}{0}

\section{Multiplicity-Free Equivariant Examples}
\label{sec:mf-examples}

The examples in this appendix are the commutative shadow of the matrix-valued
theory of Section~\ref{sec:mult}, worked out for symmetric-power
representations. Section~\ref{sec:mult-su3-example} treated a genuinely
matrix-valued case, where the intertwiner algebra is noncommutative and the
feasibility problem is semidefinite. Here the conjugation representation is
multiplicity--free, so every multiplicity space is one-dimensional, the
enumerators are scalars, and the semidefinite feasibility problem of
Theorem~\ref{thm:intrinsic-sdp-feasibility} collapses to the linear program of
Theorem~\ref{thm:intrinsic-lp-feasibility}. This is the regime in which the
whole apparatus becomes hand-computable, and in which --- in contrast with the
matrix-valued \((2,2)\) example, whose feasible enumerator at \(K=5\) is
\emph{not} unique --- the feasible enumerator is typically forced uniquely by
the linear constraints.

The symmetric powers \(V=\Sym^n(\CC^q)\) are the natural home of this regime:
their conjugation representations are multiplicity--free, and, as we show below,
they correspond exactly to permutation-invariant qudit codes with intrinsic
depth equal to ordinary qudit distance. Through that correspondence the
intrinsic linear program yields extremality certificates for concrete
permutation-invariant codes, and we record three: the four-qubit
\(\lcode 5,2,2\rcode_{\SU2}\), the seven-qubit \(\lcode 8,2,3\rcode_{\SU2}\), and
the three-qutrit \(\lcode 10,2,2\rcode_{\SU3}\) codes.

\subsection{Permutation--Invariant Qudit Codes and Symmetric-Power Intrinsic Codes}

The symmetric-power family is the bridge between the abstract intrinsic depth of
the linear program and physical qudit distance. Its multiplicity-freeness places
it in the linear regime of Appendix~\ref{sec:mf-macwilliams}, while the
distance-preserving correspondence established below makes the resulting bounds
statements about honest permutation-invariant codes.

Let $q\ge2$ and $n\ge1$, and set $ V = \Sym^n(\CC^q)$.
Then $V$ is the irreducible $\SU q$ representation of highest weight
\[
(n,0,\dots,0), \numberthis
\]
and the conjugation representation on $\L(V)$ is multiplicity--free.

\begin{proposition}
\label{prop:sym-power-mult-free}
Let $ V=\Sym^n(\CC^q) $ be the irreducible $\SU q$ representation of highest weight $ (n,0,\dots,0) $, and let $\theta$ be the highest root of $\mathfrak{sl}_q$.
Then the conjugation representation on $ \L(V)\cong V\otimes V^* $ is multiplicity--free. More precisely,
\[
\L(V)\cong \bigoplus_{r=0}^{n} \E_r,
\qquad
\E_r \cong V_{r\theta}, \numberthis
\]
with each irreducible summand occurring with multiplicity one. Here, for
$q\ge3$, $r\theta=(r,0,\dots,0,r)$ in Dynkin-label notation, while for $q=2$
the module $V_{r\theta}$ has Dynkin label $2r$, i.e.\ it is the spin-$r$
representation of $\SU 2$.
\end{proposition}

Consequently, the intrinsic enumerator formalism for
$\Sym^n(\CC^q)$ lies in the multiplicity--free equivariant setting of
\Cref{sec:mf-macwilliams}. In particular, the intrinsic enumerators are
scalar-valued and the corresponding feasibility problem is a linear program.
\begin{proof}
Since $V=\Sym^n(\CC^q)$ has highest weight $(n,0,\dots,0)$, its dual
$V^*$ has highest weight $(0,\dots,0,n)$. Hence
\[
\L(V)\cong V\otimes V^*
\cong
(n,0,\dots,0)\otimes(0,\dots,0,n). \numberthis
\]
A standard Littlewood--Richardson rule computation \cite{fultonharris} shows that this tensor
product is multiplicity--free and decomposes as
\[
(n,0,\dots,0)\otimes(0,\dots,0,n)
\cong
\bigoplus_{r=0}^{n}(r,0,\dots,0,r). \numberthis
\]
This proves the stated decomposition of $\L(V)$.
The final claim now follows immediately from the general theory of
\Cref{sec:mf-macwilliams}.
\end{proof}

Thus symmetric-power intrinsic codes form a large natural family lying entirely
in the LP regime.
Their intrinsic error sectors are indexed by the irreducible summands
\[
\E_r\cong V_{r\theta},
\qquad
0\le r\le n, \numberthis
\]
appearing in $\L(\Sym^n(\CC^q))$.

There is also a canonical $\SU q$--equivariant embedding
\[
\iota_n : \Sym^n(\CC^q)  \hookrightarrow (\CC^q)^{\otimes n} \numberthis
\]
obtained by symmetrizing tensors.
Its image is precisely the permutation--invariant subspace of
$(\CC^q)^{\otimes n}$. Thus any permutation--invariant $n$-qudit code $ \C $ may equivalently be regarded as a subspace
\[
\mathcal I := \iota_n^{-1}(\mathcal C) \subseteq V, \numberthis
\]
that is, as a subspace of the $\SU q$ irrep $(n,0,\dots,0)$.
Conversely, every subspace $\mathcal I\subseteq V$ determines a
permutation--invariant qudit code
\[
\mathcal C := \iota_n(\mathcal I)\subseteq \Sym^n(\CC^q). \numberthis
\]

The remaining point is to compare the two distance parameters. We do this
through an exact filtration identity, which yields both directions of the
comparison without a hidden surjectivity assumption.

Let $H=(\CC^q)^{\otimes n}$, let $\iota:V\hookrightarrow H$ be the symmetric
embedding, and let
\[
\Phi(E):=\iota^\dagger E\,\iota\in\operatorname{End}(V) \numberthis
\]
be the compression to the symmetric subspace. For $0\le t\le n$, let
$F_t\subseteq\operatorname{End}(H)$ denote the span of all physical errors supported on at
most $t$ sites.

\begin{theorem}[Symmetric-power error filtration]
\label{thm:sympower-filtration}
For every $0\le t\le n$,
\[
\Phi(F_t)=\bigoplus_{r=0}^{t}\E_r,
\qquad
\E_r\cong V_{r\theta}.
\]
\end{theorem}

\begin{proof}
\emph{Forward inclusion.} On one site
$\L(\CC^q)\cong\mathbf 1\oplus\Ad$, and the adjoint has highest weight
$\theta$. An error supported on at most $t$ sites therefore lies in a sum of
tensor products containing at most $t$ adjoint factors, and every irreducible
constituent of $\Ad^{\otimes s}$ has highest weight at most $s\theta$ in
dominance order; in particular $V_{r\theta}$ cannot occur for $r>s$. Since
$\Phi$ is $\SU q$--equivariant,
\[
\Phi(F_t)\subseteq\bigoplus_{r=0}^{t}\E_r. \numberthis
\]

\emph{Reverse inclusion.} Let $x=\ket 1\bra q$, a highest-root vector with
$x^2=0$, and set $J_x=\sum_{s=1}^{n}x^{(s)}$. Because $x^2=0$,
\[
J_x^{\,r}=r!\sum_{\substack{S\subseteq[n]\\|S|=r}}\ \prod_{s\in S}x^{(s)}, \numberthis
\]
a linear combination of physical errors of weight exactly $r$, so
$J_x^{\,r}\in F_r$. Moreover $J_x$ commutes with $S_n$, hence preserves the
symmetric subspace $V$, so $J_x^{\,r}|_V=\Phi(J_x^{\,r})\in\Phi(F_r)$. It is
nonzero: it maps $\ket q^{\otimes n}$ to a nonzero multiple of the symmetric
occupation state with $r$ entries equal to $1$ and $n-r$ entries equal to $q$.
Under the conjugation action $J_x^{\,r}|_V$ is a highest-weight vector of weight
$r\theta$, so the $\SU q$--submodule it generates is a copy of $V_{r\theta}$;
by multiplicity-freeness that copy is exactly $\E_r$. Hence
$\E_r\subseteq\Phi(F_r)\subseteq\Phi(F_t)$ for $r\le t$, and combining with the
forward inclusion gives equality.
\end{proof}

\begin{lemma}
\label{lem:symmetric-depth-distance}
Let
\[
V=\Sym^n(\CC^q)\subset (\CC^q)^{\otimes n}. \numberthis
\]
For any code subspace $\mathcal C\subseteq V$, the ordinary qudit distance of
$\mathcal C$ as a subspace of $(\CC^q)^{\otimes n}$ is equal to its intrinsic
depth when $\mathcal C$ is regarded as a subspace of the $\SU q$
representation $V$.
\end{lemma}

\begin{proof}
Let $P$ be the projector onto $\mathcal C$. Since $\mathcal C\subseteq V$, the
range of $P$ lies in the symmetric subspace, so for every physical operator
$E$,
\[
PEP=P\,\Phi(E)\,P. \numberthis
\]
Hence the Knill--Laflamme conditions for all physical errors of weight $<d$
are equivalent to the Knill--Laflamme conditions on $\Phi(F_{d-1})$. By
Theorem~\ref{thm:sympower-filtration}, $\Phi(F_{d-1})=\bigoplus_{r<d}\E_r$, so
these are exactly the intrinsic Knill--Laflamme conditions on the sectors
$\E_r$ with $r<d$. Therefore ordinary qudit distance at least $d$ is equivalent
to intrinsic depth at least $d$, and the two parameters coincide.
\end{proof}

\begin{proposition}
\label{prop:PI-qudit-intrinsic-strong}
Fix integers $q\ge2$ and $n\ge1$, and let
\[
V=\Sym^n(\CC^q) \numberthis
\]
be the $\SU q$ irreducible representation of highest weight $(n,0,\dots,0)$. Then the symmetric embedding
\[
\iota_n:V \xrightarrow{\;\sim\;} \Sym^n(\CC^q)\subset (\CC^q)^{\otimes n} \numberthis 
\]
induces a dimension-preserving and distance-preserving bijection between
\begin{enumerate}
\item intrinsic quantum codes in $V$, and
\item permutation--invariant quantum codes on $n$ qudits of local dimension $q$.
\end{enumerate}
Under this bijection, intrinsic depth is equal to ordinary qudit distance.
In particular, for every $K$ and $d$, intrinsic $\lcode N,K,d \rcode_{\SU{q}}$ codes in $V$ (with $N=\dim V$) are
in one-to-one correspondence with permutation--invariant $((n,K,d))_q$ codes.
\end{proposition}

\begin{proof}
The map $\iota_n$ is an $\SU q$--equivariant linear isomorphism from $V$ onto
the permutation--invariant subspace $\Sym^n(\CC^q)$.
Hence it induces inverse bijections
\[
\mathcal I \mapsto \iota_n(\mathcal I),
\qquad
\mathcal C \mapsto \iota_n^{-1}(\mathcal C), \numberthis 
\]
between subspaces of $V$ and permutation--invariant subspaces of
$(\CC^q)^{\otimes n}$.
These maps preserve dimension.
By Lemma~\ref{lem:symmetric-depth-distance}, they also preserve the distance
parameter, since intrinsic depth in $V$ equals ordinary qudit distance in
$(\CC^q)^{\otimes n}$.
\end{proof}

For $q=2$, the representation $V=\Sym^n(\CC^2)$ is the spin-$n/2$
irreducible representation of $\SU 2$.
Thus \Cref{prop:PI-qudit-intrinsic-strong} recovers the Dicke Bootstrap correspondence  between
permutation--invariant qubit codes and single-irrep $\SU 2$ codes discussed in \cite{us2}.

\subsection{Reading the examples: the linear program and its extremality certificates}
\label{subsec:mf-lp-template}

The three examples below follow a single template, which we state once here and
then instantiate.

\paragraph{Normalized enumerators.}
For a code of dimension \(K\) in \(V=\Sym^n(\CC^q)\), with \(N=\dim V\), write
\begin{equation}\label{eq:mf-normalized-enum}
\widetilde A_r:=\frac{N}{K^2}A_r,
\qquad
\widetilde B_r:=\frac{N}{K}B_r,
\qquad 0\le r\le n.
\end{equation}
In this normalization \(\widetilde A_0=\widetilde B_0=1\); the intrinsic
MacWilliams identity \(\Bvec=M\Avec\) becomes
\[
\widetilde B=KM\widetilde A;
\]
and each detected sector \(r\) contributes the coding-theoretic equality
\(\widetilde A_r=\widetilde B_r\) (Theorem~\ref{thm:matrix-KL-equality}). The
scalar linear program of Theorem~\ref{thm:intrinsic-lp-feasibility} is therefore
cut out, for a fixed dimension \(K\) and detected set \(\mathcal D\), by
\begin{equation}\label{eq:mf-lp-constraints}
0\le\widetilde A_r\le\widetilde B_r\ (0\le r\le n),
\qquad
\widetilde B=KM\widetilde A,
\qquad
\widetilde A_r=\widetilde B_r\ (r\in\mathcal D).
\end{equation}

\paragraph{Uniqueness by elimination.}
Each example determines a \emph{unique} feasible enumerator by the same
argument, which we write out in full only for the first example
(Proposition~\ref{prop:522-unique-feasible}) and invoke thereafter. Eliminate
\(\widetilde B\) via \(\widetilde B=KM\widetilde A\); impose the detection
equalities \(\widetilde A_r=\widetilde B_r\) for \(r\in\mathcal D\) using the
corresponding rows of \(M\); combine with the normalization
\(\sum_{r}\widetilde A_r=K\) (which is \(\widetilde A_0=1\) together with
\(\sum_{r\ge1}\widetilde A_r=K-1\)). The result is a linear system in the
undetected weights whose right-hand side is a nonpositive combination of
nonnegative variables; nonnegativity then forces those weights to vanish, which
pins down the top weight and hence the whole enumerator. The uniqueness of the
feasible point is the scalar counterpart of the non-uniqueness seen in the
matrix-valued \((2,2)\) example of Section~\ref{sec:mult-su3-example}.

\paragraph{Extremality.}
In each example the code is optimal within its permutation-invariant family, by
three infeasibility checks against \eqref{eq:mf-lp-constraints}: (i) for smaller
length \(n\) the fixed-\(K\), fixed-\(\mathcal D\) system is infeasible, so no
shorter code with the target parameters exists; (ii) the fixed-\((K{+}1)\)
system is infeasible, so by Lemma~\ref{lem:subcode-detection} no code of larger
dimension detects \(\mathcal D\); and (iii) enlarging the detected set (raising
the target distance) is infeasible at the given \(n,K\), so the distance cannot
be increased. Under Proposition~\ref{prop:PI-qudit-intrinsic-strong} each
verdict transfers verbatim to permutation-invariant qudit codes. We state the
resulting extremality conclusion in each case and omit the routine
re-instantiation of these three checks.

\subsection{A First Extremal Example: the Intrinsic $\lcode 5,2,2 \rcode_{\SU 2}$ Code}
\label{sec:su2-example}

The smallest nontrivial case is \(j=2\), i.e.\ \(V_2\cong\Sym^4(\CC^2)\), so that
intrinsic codes in \(V_2\) correspond to permutation-invariant codes on four
qubits (Proposition~\ref{prop:PI-qudit-intrinsic-strong}). Here
\[
\L(V_2)\cong W_0\oplus W_1\oplus W_2\oplus W_3\oplus W_4, \numberthis
\]
and by \cref{cor:M-6j} the unnormalized intrinsic MacWilliams matrix is
\[
\renewcommand{\arraystretch}{1.2}
M =
\begin{pmatrix}
\tfrac{1}{5} & \tfrac{1}{5} & \tfrac{1}{5} & \tfrac{1}{5} & \tfrac{1}{5}\\
\tfrac{3}{5} & \tfrac{1}{2} & \tfrac{3}{10} & 0 & -\tfrac{2}{5}\\
1 & \tfrac{1}{2} & -\tfrac{3}{14} & -\tfrac{4}{7} & \tfrac{2}{7}\\
\tfrac{7}{5} & 0 & -\tfrac{4}{5} & \tfrac{1}{2} & -\tfrac{1}{10}\\
\tfrac{9}{5} & -\tfrac{6}{5} & \tfrac{18}{35} & -\tfrac{9}{70} & \tfrac{1}{70}
\end{pmatrix}. \numberthis
\]

Consider the two-dimensional subspace
\[
\mathcal C
=
\operatorname{span}\left\{
\tfrac{1}{\sqrt2}(\ket0+\ket4),\;
\ket2
\right\}
\subset V_2, \numberthis
\]
where \(\{\ket0,\dots,\ket4\}\) is the standard weight basis of \(V_2\). This is
the intrinsic \(\lcode 5,2,2 \rcode_{\SU 2}\) code of
\cite{usIntrinsicCodes}; under \(V_2\cong\Sym^4(\CC^2)\) it is the four-qubit
permutation-invariant code, with distance \(d=2\)
\cite{RuskaiPRL,ouyangPI,Hagiwara2020FourQubitDeletion,
Nakayama2020SingleDeletion,Ouyang2021PermutationInvariant,422codePRL}.

\begin{proposition}
\label{prop:522-enumerator}
For the intrinsic \(\lcode 5,2,2 \rcode_{\SU 2}\) code, in the normalization
\eqref{eq:mf-normalized-enum} with \(K=2\) and \(N=\dim V_2=5\),
\[
\widetilde A=\left(1,\,0,\,0,\,0,\,\tfrac32\right),
\qquad
\widetilde B=\left(1,\,0,\,\tfrac{20}{7},\,\tfrac52,\,\tfrac{51}{14}\right). \numberthis
\]
In particular \(\widetilde A_1=\widetilde B_1=0\), so the code detects the
adjoint sector \(W_1\).
\end{proposition}

\begin{proof}
Either compute the enumerators directly from the projector \(P\) using a basis
of spherical tensors and normalize by \eqref{eq:mf-normalized-enum}, or invoke
Proposition~\ref{prop:522-unique-feasible} below, which shows this is the unique
feasible enumerator for \(K=2\), \(d=2\).
\end{proof}

\begin{proposition}
\label{prop:522-unique-feasible}
For \(j=2\), \(K=2\), and \(d=2\) (detected set \(\mathcal D=\{1\}\)), the linear
program \eqref{eq:mf-lp-constraints} has the unique feasible point
\[
\widetilde A=\left(1,\,0,\,0,\,0,\,\tfrac32\right),
\qquad
\widetilde B=\left(1,\,0,\,\tfrac{20}{7},\,\tfrac52,\,\tfrac{51}{14}\right). \numberthis
\]
\end{proposition}

\begin{proof}
This is the template elimination of Section~\ref{subsec:mf-lp-template},
written out in full. Eliminate \(\widetilde B\) by
\(\widetilde B=2M\widetilde A\). With \(\widetilde A_0=1\) the normalization is
\[
\widetilde A_1+\widetilde A_2+\widetilde A_3+\widetilde A_4=\tfrac32, \numberthis
\]
and the adjoint-detection equality \(\widetilde A_1=\widetilde B_1\), using the
second row of \(M\), reads
\[
\widetilde A_1
=\tfrac65+\widetilde A_1+\tfrac35\widetilde A_2-\tfrac45\widetilde A_4,
\qquad\text{i.e.}\qquad
\widetilde A_4=\tfrac32+\tfrac34\widetilde A_2. \numberthis
\]
Substituting into the normalization gives
\[
\widetilde A_1+\widetilde A_3+\tfrac74\widetilde A_2=0. \numberthis
\]
Since \(\widetilde A_1,\widetilde A_2,\widetilde A_3\ge0\), all three vanish,
whence \(\widetilde A_4=\tfrac32\) and
\(\widetilde A=(1,0,0,0,\tfrac32)\). Applying \(\widetilde B=2M\widetilde A\)
gives the stated \(\widetilde B\), and the feasible point is unique.
\end{proof}

By the extremality template of Section~\ref{subsec:mf-lp-template}, the
four-qubit \(((4,2,2))_2\) code corresponding to \(\lcode 5,2,2\rcode_{\SU2}\) is
optimal among permutation-invariant qubit codes: no \(((n,2,2))_2\) code exists
for \(n<4\); no \(((4,K,2))_2\) code exists with \(K>2\); and no \(((4,2,d))_2\)
code exists with \(d>2\). It is moreover the unique feasible point of its linear
program.

\subsection{An Extremal Example with Distance 3: the Intrinsic $\lcode 8,2,3 \rcode_{\SU 2}$ Code}
\label{sec:823-example}

The next case is \(j=\tfrac72\), i.e.\ \(V_{7/2}\cong\Sym^7(\CC^2)\),
corresponding to permutation-invariant codes on seven qubits. Here
\(\L(V_{7/2})\cong\bigoplus_{r=0}^{7}W_r\), and the unnormalized intrinsic
MacWilliams matrix is
\[
M=
\begin{pmatrix}
\frac18 & \frac18 & \frac18 & \frac18 & \frac18 & \frac18 & \frac18 & \frac18 \\[4pt]
\frac38 & \frac{59}{168} & \frac{17}{56} & \frac{13}{56} & \frac{23}{168} &
\frac{1}{56} & -\frac18 & -\frac{7}{24} \\[6pt]
\frac58 & \frac{85}{168} & \frac{7}{24} & \frac{5}{168} &
-\frac{5}{24} & -\frac{55}{168} & -\frac{5}{24} & \frac{7}{24} \\[6pt]
\frac78 & \frac{13}{24} & \frac{1}{24} & -\frac{31}{88} &
-\frac{101}{264} & \frac{1}{88} & \frac{119}{264} & -\frac{49}{264} \\[6pt]
\frac98 & \frac{23}{56} & -\frac38 & -\frac{303}{616} &
\frac18 & \frac{309}{616} & -\frac38 & \frac{7}{88} \\[6pt]
\frac{11}{8} & \frac{11}{168} & -\frac{121}{168} & \frac{1}{56} &
\frac{103}{168} & -\frac{363}{728} & \frac{53}{312} & -\frac{7}{312} \\[6pt]
\frac{13}{8} & -\frac{13}{24} & -\frac{13}{24} & \frac{221}{264} &
-\frac{13}{24} & \frac{53}{264} & -\frac{1}{24} & \frac{1}{264} \\[6pt]
\frac{15}{8} & -\frac{35}{24} & \frac78 & -\frac{35}{88} &
\frac{35}{264} & -\frac{35}{1144} & \frac{5}{1144} & -\frac{1}{3432}
\end{pmatrix}. \numberthis
\]

Consider the two-dimensional subspace \(\mathcal C\) spanned by
\begin{align}
    &\tfrac{\sqrt{15}}{8}\ket0
+\tfrac{\sqrt7}{8}\ket2
+\tfrac{\sqrt{21}}{8}\ket4
-\tfrac{\sqrt{21}}{8}\ket6, \\
&\tfrac{\sqrt{15}}{8}\ket7
+\tfrac{\sqrt7}{8}\ket5
+\tfrac{\sqrt{21}}{8}\ket3
-\tfrac{\sqrt{21}}{8}\ket1 \numberthis
\end{align}
in the standard weight basis \(\{\ket0,\dots,\ket7\}\) of \(V_{7/2}\). This is
the intrinsic \(\lcode 8,2,3 \rcode_{\SU 2}\) code; under
\(V_{7/2}\cong\Sym^7(\CC^2)\) it is the seven-qubit permutation-invariant code,
with distance \(d=3\), and equivalent realizations appear in
\cite{2004permutation,gross1,us1}.

\begin{proposition}
\label{prop:823-enumerator}
For the intrinsic \(\lcode 8,2,3 \rcode_{\SU 2}\) code, in the normalization
\eqref{eq:mf-normalized-enum} with \(K=2\) and \(N=8\),
\[
\widetilde A =\left(1,0,0,0,0,0,3,0\right),
\qquad
\widetilde B =
\left(1,0,0,\tfrac{49}{11},0,\tfrac{49}{13},3,\tfrac{540}{143}\right). \numberthis
\]
In particular \(\widetilde A_1=\widetilde B_1=0\) and
\(\widetilde A_2=\widetilde B_2=0\), so the code detects the sectors \(W_1\) and
\(W_2\).
\end{proposition}

\begin{proof}
Compute directly from the code projector, or invoke
Proposition~\ref{prop:823-unique-feasible}.
\end{proof}

\begin{proposition}
\label{prop:823-unique-feasible}
For \(j=\tfrac72\), \(K=2\), and \(d=3\) (detected set \(\mathcal D=\{1,2\}\)),
the linear program \eqref{eq:mf-lp-constraints} has the unique feasible point of
Proposition~\ref{prop:823-enumerator}.
\end{proposition}

\begin{proof}
We follow the template elimination of Section~\ref{subsec:mf-lp-template}; the
one new feature is that two detection equalities give two expressions for the
same weight. With \(\widetilde A_0=1\) the normalization is
\(\sum_{r=1}^{7}\widetilde A_r=3\). Using the second and third rows of \(M\), the
detection equalities \(\widetilde A_1=\widetilde B_1\) and
\(\widetilde A_2=\widetilde B_2\) each solve for \(\widetilde A_6\):
\begin{align}
\widetilde A_6 &= 3
     + \tfrac{17}{7}\widetilde A_2 + \tfrac{13}{7}\widetilde A_3
     + \tfrac{23}{21}\widetilde A_4 + \tfrac{1}{7}\widetilde A_5
     - \tfrac{25}{21}\widetilde A_1 - \tfrac{7}{3}\widetilde A_7,
\label{eq:A6fromB1} \\
\widetilde A_6 &= 3
     + \tfrac{17}{7}\widetilde A_1 + \tfrac{1}{7}\widetilde A_3
     - \widetilde A_4 - \tfrac{11}{7}\widetilde A_5 - \widetilde A_2 + \tfrac{7}{5}\widetilde A_7. \numberthis
\label{eq:A6fromB2}
\end{align}
The normalization together with \eqref{eq:A6fromB1} gives
\begin{equation}
\widetilde A_1 = 18\widetilde A_2 + 15\widetilde A_3 + 11\widetilde A_4 + 6\widetilde A_5 - 7\widetilde A_7, \numberthis
\label{eq:A1}
\end{equation}
while equating \eqref{eq:A6fromB1} and \eqref{eq:A6fromB2} gives
\begin{equation}
540\widetilde A_2 + 460\widetilde A_3 + 330\widetilde A_4 + 175\widetilde A_5 = 189\widetilde A_7. \numberthis
\label{eq:linrel}
\end{equation}
Substituting \eqref{eq:linrel} into \eqref{eq:A1} yields
\begin{equation}
\widetilde A_1 = -2\widetilde A_2 - \tfrac{55}{27}\widetilde A_3 - \tfrac{11}{9}\widetilde A_4 - \tfrac{13}{27}\widetilde A_5. \numberthis
\label{eq:A1neg}
\end{equation}
As every term on the right is nonpositive while
\(\widetilde A_1,\dots,\widetilde A_5\ge0\), both sides vanish, so
\(\widetilde A_1=\dots=\widetilde A_5=0\); then \eqref{eq:linrel} gives
\(\widetilde A_7=0\) and \eqref{eq:A6fromB1} gives \(\widetilde A_6=3\). Thus
\(\widetilde A=(1,0,0,0,0,0,3,0)\), and \(\widetilde B=2M\widetilde A\) gives the
stated \(\widetilde B\). The feasible point is unique.
\end{proof}

\begin{remark}
\label{rmk:823-degeneracy}
Because the feasible point is unique, every permutation-invariant \(((7,2,3))\)
code has this same intrinsic enumerator. In particular the two Pollatsek--Ruskai
codes and the further seven-qubit permutation-invariant distance-\(3\)
code~\cite{2004permutation,us1} are indistinguishable at the level of the
intrinsic linear program.
\end{remark}

By the extremality template, the seven-qubit \(((7,2,3))\) code corresponding to
\(\lcode 8,2,3\rcode_{\SU2}\) is optimal among permutation-invariant qubit
codes: no \(((n,2,3))\) code exists for \(n<7\); no \(((7,K,3))\) code exists
with \(K>2\); and no \(((7,2,d))\) code exists with \(d>3\). It is the unique
feasible point of its linear program.

\subsection{A Qutrit Example: the Intrinsic $\lcode 10,2,2 \rcode_{\SU 3}$ Code}
\label{sec:su3-example}

The final example is genuinely qutrit: \(V=\Sym^3(\CC^3)\), of dimension \(10\),
corresponding to permutation-invariant codes on three qutrits
(Proposition~\ref{prop:PI-qudit-intrinsic-strong}). The multiplicity-free
decomposition is
\[
\L(V)\cong(0,0)\oplus(1,1)\oplus(2,2)\oplus(3,3), \numberthis
\]
and the unnormalized intrinsic MacWilliams matrix is
\[
\renewcommand{\arraystretch}{1.3}
M=
\begin{pmatrix}
\frac{1}{10} & \frac{1}{10} & \frac{1}{10} & \frac{1}{10}\\
\frac{4}{5} & \frac{3}{5} & \frac{4}{15} & -\frac{1}{5}\\
\frac{27}{10} & \frac{9}{10} & -\frac{47}{70} & \frac{9}{70}\\
\frac{32}{5} & -\frac{8}{5} & \frac{32}{105} & -\frac{1}{35} 
\end{pmatrix}. \numberthis
\]

An orthonormal weight basis of \(V\) is given by the occupation-number vectors
\(\ket{a,b,c}\) with \(a+b+c=3\), the normalized symmetric sums of computational
basis states with \(a\) copies of \(0\), \(b\) of \(1\), and \(c\) of \(2\).
Consider
\[
\mathcal C
=
\operatorname{span}\left\{
\tfrac{1}{\sqrt3}\bigl(\ket{3,0,0}+\ket{0,3,0}+\ket{0,0,3}\bigr),
\;
\ket{1,1,1}
\right\}
\subset V. \numberthis
\]
where
\[
\ket{1,1,1}=\tfrac{1}{\sqrt6}
\bigl(\ket{012}+\ket{021}+\ket{102}+\ket{120}+\ket{201}+\ket{210}\bigr). \numberthis
\]
This is the intrinsic \(\lcode 10,2,2 \rcode_{\SU 3}\) code; it is exactly the
permutation-invariant subcode of the three-qutrit stabilizer code with
generators \(XXX\) and \(ZZZ\).

\begin{proposition}
\label{prop:322-enumerator}
For the intrinsic \(\lcode 10,2,2 \rcode_{\SU 3}\) code, in the normalization
\eqref{eq:mf-normalized-enum} with \(K=2\) and \(N=10\),
\[
\widetilde A=\left(1,\,0,\,0,\,4\right),
\qquad
\widetilde B=\left(1,\,0,\,\tfrac{45}{7},\,\tfrac{88}{7}\right). \numberthis
\]
In particular \(\widetilde A_1=\widetilde B_1=0\), so the code detects the
adjoint sector \((1,1)\).
\end{proposition}

\begin{proof}
Compute directly from the code projector, or invoke
Proposition~\ref{prop:322-unique-feasible}.
\end{proof}

\begin{proposition}
\label{prop:322-unique-feasible}
For \(V=\Sym^3(\CC^3)\), \(K=2\), and \(d=2\) (detected set
\(\mathcal D=\{1\}\)), the linear program \eqref{eq:mf-lp-constraints} has the
unique feasible point of Proposition~\ref{prop:322-enumerator}.
\end{proposition}

\begin{proof}
Template elimination of Section~\ref{subsec:mf-lp-template}. With
\(\widetilde A_0=1\) the normalization is
\(\widetilde A_1+\widetilde A_2+\widetilde A_3=4\), and the adjoint-detection
equality \(\widetilde A_1=\widetilde B_1\), via the second row of \(M\), gives
\(\widetilde A_3=4+\tfrac12\widetilde A_1+\tfrac43\widetilde A_2\). Substituting
into the normalization yields
\(\tfrac32\widetilde A_1+\tfrac73\widetilde A_2=0\); nonnegativity forces
\(\widetilde A_1=\widetilde A_2=0\) and \(\widetilde A_3=4\, . \) Then
\(\widetilde B=2M\widetilde A\) gives the stated \(\widetilde B\), and the
feasible point is unique.
\end{proof}

By the extremality template, the three-qutrit \(((3,2,2))_3\) code corresponding
to \(\lcode 10,2,2\rcode_{\SU3}\) is optimal among permutation-invariant qutrit
codes: no \(((n,2,2))_3\) code exists for \(n<3\); no \(((3,K,2))_3\) code exists
with \(K>2\); and no \(((3,2,d))_3\) code exists with \(d>2\). It is the unique
feasible point of its linear program.

\makeatletter
\renewcommand{\theequation}{C\arabic{equation}}
\renewcommand{\thetable}{C\arabic{table}}
\renewcommand{\thefigure}{C\arabic{figure}}
\renewcommand{\thelemma}{C\arabic{lemma}}
\renewcommand{\theremark}{C\arabic{remark}}
\renewcommand{\theproposition}{C\arabic{proposition}}
\renewcommand{\thecorollary}{C\arabic{corollary}}
\renewcommand{\thetheorem}{C\arabic{theorem}}
\setcounter{table}{0}
\setcounter{figure}{0}
\setcounter{lemma}{0}
\setcounter{remark}{0}
\setcounter{corollary}{0}
\setcounter{proposition}{0}
\setcounter{theorem}{0}
\setcounter{equation}{0}

\section{Exact positive dual certificate for the depth-$2$ bound in $V=(2,2)$}
\label{app:22-certificate}

This appendix records the positive dual certificate used in the proof of
Proposition~\ref{prop:22-Kle5}. All entries are exact; the certificate has been
verified in exact arithmetic, including the identity
\eqref{eq:22-supporting-identity}, the positive definiteness of every \(Y_\xi\),
and the weighted orthogonality \(MDM^{T}=D\).

\paragraph{Coordinate order.}
We use the $33$-dimensional enumerator coordinate system obtained by listing the
isotypic sectors of \(\L(V)\) in the order
\[
(0,0);\ (1,1);\ (2,2);\ (0,3);\ (3,0);\ (4,1);\ (1,4);\ (3,3);\
(2,5);\ (5,2);\ (6,0);\ (0,6);\ (4,4),
\]
and, within a sector of multiplicity \(m_\xi\), flattening the block entries as
\((\alpha,\beta)=(1,1),(1,2),\dots,(1,m_\xi),(2,1),\dots,(m_\xi,m_\xi)\). In this
basis the coefficient transform \(M\), in the twirl convention
\eqref{eq:twirl-matrix-unit-ordered} used throughout the main text, satisfies
\(\mathrm{vec}(B)=M\,\mathrm{vec}(A)\) and \(MDM^{T}=D\), where \(D\) is diagonal
with \(d_\xi\) repeated \(m_\xi^2\) times.

\paragraph{Dual blocks.}
Let \(S=\{(1,1),(2,2),(0,3),(3,0),(4,1),(1,4),(3,3)\}\) and
\(Z=\operatorname{diag}\!\left(\tfrac{9}{35},1\right)\). Define
\[
Y_{(1,1)}=
\begin{pmatrix}
\frac{1346}{819} & -\frac{8\sqrt{105}}{455}\\[4pt]
-\frac{8\sqrt{105}}{455} & \frac{10}{13}
\end{pmatrix},
\qquad
Y_{(3,3)}=
\begin{pmatrix}
\frac{133}{153} & -\frac{40}{357}\\[4pt]
-\frac{40}{357} & \frac{25}{119}
\end{pmatrix},
\]
\[
Y_{(2,2)}=
\begin{pmatrix}
\frac{39776}{44415} & \frac{8\sqrt{4935}}{1269} & -\frac{32\sqrt{35}}{8883}\\[5pt]
\frac{8\sqrt{4935}}{1269} & \frac{64}{63} & \frac{160\sqrt{141}}{8883}\\[5pt]
-\frac{32\sqrt{35}}{8883} & \frac{160\sqrt{141}}{8883} & \frac{12800}{8883}
\end{pmatrix},
\]
\[
Y_{(4,1)}=
\begin{pmatrix}
\frac{16}{21} & -\frac{8\sqrt6}{63}\\[4pt]
-\frac{8\sqrt6}{63} & \frac{32}{63}
\end{pmatrix},
\qquad
Y_{(1,4)}=
\begin{pmatrix}
\frac{416}{441} & -\frac{8\sqrt6}{147}\\[4pt]
-\frac{8\sqrt6}{147} & \frac{16}{49}
\end{pmatrix},
\]
and \(Y_{(0,3)}=Y_{(3,0)}=\tfrac{64}{45}\). Each \(Y_\xi\) is positive definite;
for instance \(\det Y_{(1,1)}=\tfrac{388}{315}\),
\(\det Y_{(4,1)}=\det Y_{(1,4)}=\tfrac{128}{441}\),
\(\det Y_{(3,3)}=\tfrac{25}{147}\), and for the \(3\times3\) block the leading
principal minors are \(\tfrac{39776}{44415}\), \(\tfrac{1996864}{2798145}\), and
\(\det Y_{(2,2)}=\tfrac{81920}{83349}\), all positive, so Sylvester's criterion
applies.

\paragraph{Supporting identity.}
Direct exact multiplication by the transform \(M\) gives the identity
\eqref{eq:22-supporting-identity}, valid for all \(A\) with
\(\mathrm{vec}(B)=M\,\mathrm{vec}(A)\); this is the supporting-hyperplane
functional used in Proposition~\ref{prop:22-Kle5}.

\bibliographystyle{IEEEtran}
\bibliography{biblio.bib}

\end{document}